\documentclass[11pt,twoside,letterpaper]{article}
\usepackage{hyperref}
\usepackage{amsmath}
\usepackage{amssymb}
\usepackage{amsfonts}
\usepackage{amsthm}
\usepackage{comment}
\usepackage{color}
\usepackage{subfig}
\usepackage{graphicx}

\def\cM{{\mathcal {M}}}

\def\cK{{\mathcal {K}}}

\def\p{\partial}

\def\R{{ \mathbb{R}}}

\def\N{{\mathbb{N}}}
\def\Z{{\mathbb{Z}}}

\def\cR{{\mathcal R}}

\def\cB{{\cal B}}

\def\cM{{\cal M}}

\def\br{{\bf r}}

\newtheorem{theorem}{Theorem}[section]
\newtheorem{lemma}[theorem]{Lemma}
\newtheorem{proposition}[theorem]{Proposition}

\newtheorem{corollary}[theorem]{Corollary}
\newtheorem{remark}[theorem]{Remark}

\theoremstyle{definition}
\newtheorem{definition}[theorem]{Definition}
\newtheorem{example}[theorem]{Example}

\newtheorem{assumption}[theorem]{Assumption}

\newtheorem*{remark*}{Remark}

\numberwithin{figure}{section} \numberwithin{equation}{section}

\title{Epidemic outbreaks in structured host populations}

\author{
Horst R. Thieme\thanks{({\tt hthieme@asu.edu})}
\\
School of Mathematical and Statistical Sciences
\\
Arizona State University,
Tempe, AZ 85287-1804, USA
}
\date{April 28, 2025}

\begin{document}

\maketitle

%%%%%%%%%%%%%%%%%%%%%%%%%%%%%%%%%%%%%%%%%%%%%%%%

\begin{abstract}
For a heterogeneous host population,
the basic reproduction number of an infectious
disease, $\cR_0$, is defined as the spectral radius
of the next generation operator (NGO).
The threshold properties of the
basic reproduction number are typically
established by imposing conditions that make $\cR_0$  an
eigenvalue of the NGO associated with a positive
eigenvector and a positive eigenfunctional
(eigenvector of the dual of the NGO). More general results can
be obtained by imposing conditions that
associate $\cR_0$ just with a positive eigenfunctional.
The next generation operator is conveniently
expressed by a measure kernel or a Feller kernel which enables the use of analytic
rather than functional analytic methods.
\end{abstract}
{\bf Keywords:} force of infection,
tight measure kernels, topologically irreducible
kernels, dominated kernels,  solid cones

\medskip

\noindent
{\bf MSC Classification:} 92D30, 92D25,  28A25, 45H05, 45M99, 47J05, 47N20

%%%%%%%%%%%%%%%%%%%%%%%%%%%%%%%%%%%%%%%%%%%%%%

\section{Introduction}
\label{sec:intro}

The emergence of Covid-19 has reignited the interest
in epidemic models of Kermack--McKen\-drick type.
See  \cite{BCCR, BoDI1, BoDI2, DeGrMa, DeMa, Diek20, DiIn, DiOt, dIM, dMI, DuMa, Eik, Gom, KrSc, LiMa, LiMaWe, LuSt-free, LuSt-bol, PoThar, SaVH, Tka1, Tka2} and the references therein. The host populations typically are heterogenous,
and it is of interest how the spread of the infection is affected by the structure of the host
population. It is one of the mathematical consequences
of host heterogeneity that the
basic reproduction number, $\cR_0$, is formulated
in a functional analytic way, as the spectral
radius of an appropriate positive bounded linear operator
on the   function  space for the force of infection.
In order to explore the threshold character of $\cR_0$,
the existing
literature typically uses that, under appropriate
assumptions, the spectral radius
is an eigenvalue of the next generation operator
associated with a positive eigenvector and a positive
eigenfunctional (eigenvector of the dual opeator).
Persistence theory \cite{SmTh, Thi24}
teaches the lesson
that it can be sufficient if the spectral
radius is just associated with a positive eigenfunctional  \cite[Thm.7]{Kar} \cite[App.2.6]{Sch59}.

It is one of the aims of this article to show how
the existence of an eigenfunctional  plays out for epidemic models in structured host populations, in
particular in getting away with less compactness
assumptions and, with the exception of one
fundamental theorem (Section \ref{subsec:Krein}), the use of real rather than
functional analysis.

\subsection{Scent of the model}
\label{subsec:model-intro}

The structure of the host population is described
by a nonempty set $\Omega$ of structural characteristics of the hosts, $x$, also called traits, which do not
change over time. Only such static traits are considered in this paper.
One of the first population structures that have been considered
appears to be associated with the spatial spread of the
epidemic, and the trait of the host is spatial location. Some of the respective models have been formulated
so generally that they apply to general structures
\cite{Die78, Thi77}.

Let us assume that there is a function
$s: \R_+ \times \Omega \to (0,1]$
such that $s (t,x) $ is the probability that
a typical host with trait $x$ that is susceptible at
time 0 (the begin of the epidemic) will still  be susceptible at time $t>0$ \cite{BoDI1}. In an epidemic scenario,
where infected hosts do not become susceptible,
$s(\cdot,x)$ is a decreasing function of $t$
and $s(0,x) =1$. Define
\begin{equation}
J(t,x) = - \ln s(t,x).
\end{equation}
Then $J(\cdot,x)$ is increasing, $J(0,x)=0$, and
\begin{equation}
\label{eq:incid-intro}
s(t,x) = e^{-J(t,x)}.
\end{equation}
$J$ is called the {\it cumulative force of infection}, cFoI,
up to time $t$. Under appropriate assumptions,
\begin{equation}
\label{eq:cum-force-intro}
J(t,x ) = \int_0^t I(r,x) dr,
\end{equation}
resulting in the differential equation
\begin{equation}
\label{eq:incidence-intro}
\p_t s(t,x) = - s(t,x) I(t,x),
\end{equation}
with a non-negative function $I: \R_+ \times \Omega
\to \R$. $I(t,x)$ is called the {\it force of infection}, FoI, exerted on
susceptible hosts of trait $x$  at time $t$ by all presently infected hosts. Cf. \cite{BuCo, DiHeBr, IaMi, Inabook, Mar}.

Some more equations will connect  $s$ and $I$
(Section \ref{sec:model}),
which we do not present here because they
are quite complex.
In the spirit of Kendrick and McKermack
\cite{KeMcK1, KeMcK2, KeMcK3},
we will assume that the infectivity of an
infected host does not only depend on its trait
but also on its infection age (a special form of
class age \cite{IaMi, Thi77}, the time that
has passed since the moment of being infected
\cite{BCCR, BoDI1, DiHeBr, Inabook, Ina23, LYMar, MaRu, Mar}). Differently from \cite{BoDI1}, we do not
assume that there is a prescribed prehistory
of the epidemic, but similarly to \cite{Die78, Ina23, RaRa, Thi77} we assume the scenario
that the disease is introduced into the population
at some time $t_0$ which is normalized to $t_0=0$.
This seems to be more appropriate for emerging diseases like Covid-19.

By its interpretation, it is suggestive
that the {cFoI} $J(t,x)$ is an increasing function of $t$. If is bounded, we can define the final size of
the cFoI,
\begin{equation}
\label{eq:final-size-intro}
w(x) = \lim_{t \to \infty} J(t,x), \qquad x \in \Omega.
\end{equation}
By (\ref{eq:incid-intro}),
\begin{equation}
\label{eq:final-size-susc-intro}
s_\infty(x) := \lim_{t\to \infty}s(t,x) = e^{-w(x)},
\qquad
x \in \Omega.
\end{equation}

\subsection{Overview}
\label{subsec:overv}
The final size of the cFoI, $w$ in (\ref{eq:final-size-intro}) satisfies
a Hammerstein equation with a measure kernel.
See Section \ref{sec:basic}, where we also introduce
the concept of a next generation operator (NG0)
and of the basic reproduction number $\cR_0$ as its
spectral radius.
In Section \ref{subsec:preview}, we
give a preview  of threshold results for $\cR_0$,
which are discussed in Section \ref{subsec:disc-intro}.
We also give the elementary proof of a fundamental threshold theorem.
In Section \ref{sec:model}, we derive integral equations of Volterra Hammerstein form for the FoI and for the cFoI.
In Section \ref{sec:existence}, following
\cite{Thi77}, we show the
existence of minimal solutions to these equations.
Minimal solutions are unique by their nature, and we argue that they are the epidemiologically relevant solutions.
In Section \ref{sec:final-size}, again following
\cite{Thi77}, we derive
 Hammerstein equations with measure kernels for the final size
of the cFoI and prove the existence of minimal
solutions.
In  Section \ref{sec:intermed-thres} and \ref{sec:semi-sep},
we show threshold theorems
for dominated and semi-separable measure kernels,
and in Section \ref{sec:metric}
 for Feller kernels in  case that $\Omega$
 is a metric space.

%%%%%%%%%%%%%%%%%%%%%%%%%%%%%%%%%%%%%%%%%
\section{Some basic concepts}
\label{sec:basic}

In this section, we get to see the equation for the final size of the cFoI, $w$ in (\ref{eq:final-size-intro}), and
we introduce
the concept of a next generation operator (NGO)
and of the basis reproduction number $\cR_0$ as its
spectral radius.
%%%%%%%%%%%%%%%%%%%%%%%%%%%%%%%%%%%%%%%%%%%%%%%%%%%%%%%
\subsection{The final size of the epidemic and measure kernels}
\label{subsec:final-size-intro}

Whereas  the equation for the cumulative force of infection $J$ is a quite complicated Volterra
Hammerstein equation,
the equation for its final size $w$ boils down to a Hammerstein equation
\begin{equation}
\label{eq:final-size-eq-intro}
w(x) - w^\circ(x) = \int_\Omega f(w(\xi)) \kappa( d \xi,x)=: F(w)(x),  \qquad x \in \Omega.
\end{equation}
Cf. \cite[(31)]{Ina23}. While $w$ is the final size of the cFoI exerted by all infected hosts,
$w^\circ$ is the final size of the cFoI exerted
by the hosts that were infected at the beginning
of the epidemic.
Here,
\begin{equation}
\label{eq:skellam-intro}
f(y) = 1 - e^{-y}, \qquad y \in \R_+.
\end{equation}
$\Omega$ is a measurable space with a $\sigma$-algebra
$\cB$, e.g., an open or closed subset of $\R^n$.
Further, $\kappa: \cB \times \Omega \to \R_+$
is a measure kernel. Let us explain.

Let $\cM(\Omega)$ be the vector space of
 real-valued measures on $\cB$ and $\cM_+(\Omega)$
the order cone of nonnegative measures.
See \cite{Thi19}\cite[Chap.33]{Thi24} for a population oriented
introduction.

Further, let $M^b(\Omega)$ be
 the Banach space of bounded measurable functions with
the supremum norm, $\|\cdot\|_\infty$,  and $M^b_+(\Omega)$ be
the closed order cone of nonnegative functions.
The final sizes $w$ and $w^\circ$ are elements of
$M^b_+(\Omega)$.

Now, $\kappa$ is a {\it measure kernel}
\cite{Thi19}\cite[Sec.13.2]{Thi24}  if
\begin{equation}
\label{eq:kernel-meas-intro}
 \kappa(\cdot, x) \in \cM_+(\Omega), \quad x \in \Omega,
\qquad
 \kappa( \omega, \cdot) \in M^b_+(\Omega), \quad
\omega \in \cB.
\end{equation}
Measure kernels are considered for a uniform
treatment of the initial-value problem version
of the model in \cite{BoDI1} and of some special cases of \cite[Sec.3]{Thi77}. Notice that
both the frameworks in \cite{BoDI1} and in \cite{Thi77} allow the consideration of continuous and discrete population structure.
We will show (Section \ref{sec:final-size}) that the final size $w$ defined by (\ref{eq:final-size-intro}) is the minimal
solution of (\ref{eq:final-size-eq-intro}),
given as the pointwise limit $w = \lim_{n \to
 \infty} w_n$ of the recursion
\begin{equation}
\label{eq:recursion-intro}
w_{n} = F(w_{n-1}) + w^\circ, \quad  n \in \N,
\qquad
w_0 = w^\circ.
\end{equation}
By induction, $(w_n)$ is a increasing sequence
of functions in $M^b_+(\Omega)$.
Cf. \cite[Sec.3]{Thi77}.

%%%%%%%%%%%%%%%%%%%%%%%%%%%%%%%%%%%%%%%%%%%%%%%%%%%%%

\subsection{Next generation operator
and basic reproduction number}

The event of no infection is represented by
the zero function.
The linearization of $F$ in (\ref{eq:final-size-eq-intro}) at the zero function
is the positive bounded linear operator $K: M^b(\Omega) \to M^b(\Omega)$ given by
\begin{equation}
\label{eq:lin-approx-intro}
(Kg)(x) = \int_\Omega g(\xi) \kappa(d \xi,x),
\qquad
x \in \Omega, \quad g \in M^b(\Omega).
\end{equation}
More precisely,
$K$ is the order-derivative of $F$
\cite{Thi19}\cite[Sec.14.2]{Thi24}:
By the mean value theorem applied to $f$ in
(\ref{eq:skellam-intro}),
since $f'$ is decreasing,
\begin{equation}
\label{eq:skellam-deriv}
f(y) \ge y f'(y) = y e^{-y}, \qquad y \in \R_+,
\end{equation}
\begin{equation}
\label{eq:order-der}
f'\big(\|g\|_\infty \big)\, Kg \le F(g) \le Kg,
\quad
g \in M^b_+(\Omega), \qquad f'(0) =1.
\end{equation}
\begin{remark}
\label{re:final-size-lin}
Even more precisely, by (\ref{eq:skellam-deriv}),
\[
F(g) (x) \ge \int_\Omega f'(g(\xi))\, g(\xi) \,\kappa(d \xi,x), \qquad x \in \Omega , \quad g \in M^b_+(\Omega).
\]
By (\ref{eq:final-size-susc-intro}), if $w$ is the final
size of the cFoI,
\[
F(w)(x) \ge \int_\Omega w(\xi) s_\infty(\xi) \kappa(d\xi,x),
\qquad x \in \Omega.
\]
\end{remark}
In our epidemiological context, $K$ plays the role of a {\em next generation operator}
(with generation to be understood in an epidemiological sense). See Section \ref{subsec:measure-kernels} for more details.

The {\em basic reproduction number of the epidemic}, $\cR_0$, is defined
as the spectral radius of $K$, which is given
by the Gelfand formula
\begin{equation}
\label{eq:spec-rad-Gel-intro}
\cR_0 = \br (K) = \inf_{n \in \N} \|K^n\|^{1/n}
=
\lim_{n\to \infty} \|K^n\|^{1/n}.
\end{equation}
Here $K^n$ is the $n$-fold composition (iterate)
of $K$ with itself and $\|K^n\|$ its operator norm.
See \cite[A.3.2]{Thi24} for a proof of the
last equality that also holds if $K$ is not linear
but only homogeneous.

 We also call $\cR_0$
the spectral radius of the measure kernel $\kappa$,
\begin{equation}
\label{eq:specrad-kernel}
\cR_0 = \br (\kappa).
\end{equation}
$\cR_0$ will turn out to be an epidemic threshold
parameter for the type of models we will consider. Cum grano salis, epidemic outbreaks can occur
if $\cR_0 >1$, and do not occur if $\cR_0 <1$.

%%%%%%%%%%%%%%%%%%%%%%%%%%%%%%%%%%%%%%%%%%%%%%%%%

\section{Preview of threshold results}
\label{subsec:preview}

In the following, we will give a glimpse of
the epidemic threshold character of the
basic reproduction number, $\cR_0$.

\subsection{A fundamental result from positive operator
theory}
\label{subsec:Krein}

It follows from the definition of a measure kernel,
(\ref{eq:kernel-meas-intro}),
that
$u: \Omega \to \R_+$,
\begin{equation}
\label{eq:order-unit}
u(x) =  \kappa(\Omega,x), \quad x \in \Omega,
\end{equation}
is an element of $M^b_+(\Omega)$.
The linear operator $K$ on $M^b(\Omega)$
given by (\ref{eq:lin-approx-intro})
is uniformly $u$-bounded, i.e.,
\begin{equation}
\label{eq:order-bounded}
Kg \le \|g\|_\infty u, \qquad g \in M^b_+(\Omega).
\end{equation}
Further, the constant function $u_1(x) =1$
is an interior point of $M^b_+(\Omega)$,
\begin{equation}
\label{eq:solid}
|g(x)| \le \|g\|_\infty u_1(x), \qquad x \in \Omega, \quad g \in M^b(\Omega).
\end{equation}
After these preparations, the following theorem
is a special case of more abstract ones for
positive bounded linear operators on
ordered Banach spaces with a solid normal cone
{\cite[Thm.7]{Kar} \cite{KrRu} \cite[App.2.6]{Sch59}\cite[Cor.11.17]{Thi24}}.

\begin{theorem}
\label{re:Krein}
If $\cR_0 >0$, there exists  a
 bounded linear positive eigenfunctional $\theta: M^b(\Omega) \to \R$ of $K$ (and of $\kappa$) such that $\theta (K g)= \cR_0 \, \theta(g) $ for all $g \in M^b(\Omega)$, $\theta(g) > 0$ if $g \in M^b_+(\Omega)$ and $g \ge \delta u$
 for some $\delta >0$.
\end{theorem}

The inequality $\theta(u) > 0$ follows from (\ref{eq:order-bounded}).
Otherwise, since $\theta$ is increasing and linear, \[
\cR_0 \theta(g) = \theta (Kg) \le \|g\|_\infty \, \theta (u) =0, \qquad g \in M^b_+(\Omega),
\]
and $\theta$ would be the zero functional.

The functional $\theta$ can rightfully be called
also an eigenfunctional of the measure kernel
$\kappa$ because it is characterized by
\begin{equation}
\label{eq:eigenfunc-kappa}
\theta (\kappa(\omega, \cdot))= \cR_0 \, \theta(\chi_\omega), \qquad \omega \in \cB,
\end{equation}
where $\chi_\omega$ is the indicator or characteristic
function of the set $\omega$, $\chi_\omega(\xi) =1 $
if $\xi \in \omega$ and $\chi_\omega(\xi) =0$ if $\xi
\in \Omega \setminus \omega$.

Generalizations of Theorem \ref{re:Krein}
to increasing bounded homogeneous  operators
can be found in \cite{Thi16} \cite[Chap.11]{Thi24}.

%%%%%%%%%%%%%%%%%%%%%%%%%%%%%%%%%%%%%%%%%%%%%%%%%%%

\subsection{A general but weak threshold result}
\label{subsec:general-thres}

Persistence theory
teaches that eigenfunctionals of the NGO are as important tools as eigenvectors \cite{SmTh, Thi24}.

\begin{theorem}
\label{re:epi-occur}
Let $w \in M^b_+(\Omega)$ be the minimal solution of $w = F(w) + w^\circ$
 which is uniquely
determined by $w^\circ \in M^b_+(\Omega)$.

\begin{itemize}

\item[(a)] Let  $\cR_0 <1$. Then there exists
a constant $c>0$ (independent of $w$ and $w^\circ$)
such that $\|w\|_\infty \le c \|w^\circ\|_\infty$.

\item[(b)] Let $\cR_0 > 1$   and $\theta$ be the eigenfunctional of $K$ from Theorem \ref{re:Krein}. If $\theta(w) > 0$ (in particular if $\theta(w^\circ) >0$),
then $\|w - w^\circ\|_\infty \ge \ln \cR_0$ and $ \inf_{\Omega} s_\infty \le 1/\cR_0$ for $s_\infty$ in (\ref{eq:final-size-susc-intro}).
\end{itemize}
\end{theorem}

Notice that there are no irreducibility or
compactness assumptions for $K$. As a trade-off,
there is not so much information as in other
threshold results (see \cite{BoDI1, Ina23} and
the Sections (\ref{subsec:dominated-intro}) to
(\ref{subsec:metric-intro}). But one
definitely sees that it makes a difference
whether $\cR_0 <1$ or $\cR_0 >1$.

This result is  fundamental for our approach and the proof is
easy (once we take Theorem \ref{re:Krein} for
granted) so that we give it right here.

\begin{proof} (a) By (\ref{eq:order-der}) and
(\ref{eq:recursion-intro}),
\[
w_{n+1} \le K w_n + w^\circ, \quad n \in \N.
\]
By induction,
\[
w_n \le \sum_{j=0}^n K^j w^\circ, \quad n \in \N.
\]
Since $\cR_0 <1$,  $K_\infty =\sum_{j=0}^\infty K^j$
exists with the series converging even in operator norm
and is a bounded linear operator on $M^b(\Omega)$,
\[
w_n(x) \le (K_\infty w^\circ)(x), \quad n \in \N,
\quad x \in \Omega.
\]
Taking pointwise limits, $w = \lim_{n\to \infty} w_n$,
\[
w \le K_\infty w^\circ
\quad \hbox{ and } \quad
\|w\|_\infty \le \|K_\infty\| \, \|w^\circ\|_\infty.
\]

(b) Assume $\theta(w) > 0$. Let $\tilde w = w - w^\circ$.
By  (\ref{eq:order-der}),
\[
\tilde w = F(w) \ge f'(\|w\|_\infty) K(w).
\]
Since $\theta$
is an eigenfunctional of $K$ associated with $\cR_0$ and $\theta $ is increasing,
\[
\theta (\tilde w) \ge f'(\|w\|_\infty)\theta K(w)= e^{-\|w\|_\infty} \cR_0 \, \theta(w) > 0.
\]
Since $F$ is increasing, again by (\ref{eq:order-der}),
\[
\tilde w \ge F(\tilde w)\ge f'(\|\tilde w\|_\infty) K(\tilde w)
\]
and, by applying $\theta$,
\[
\theta(\tilde w) \ge f'(\|\tilde w\|_\infty) \cR_0 \,
\theta (\tilde w).
\]
Since $\theta(\tilde w) > 0$,
$ f'(\|\tilde w\|_\infty)\le 1/\cR_0$
and $\|\tilde w\|_\infty
\ge \ln \cR_0.$
By (\ref{eq:final-size-susc-intro}),
\[
\inf_{x\in \Omega} s_\infty(x)
=
e^{-\|w\|_\infty} \le f'(\|\tilde w\|_\infty)\le 1/\cR_0.
\]
Since $\theta$ is increasing (linear and nonnegative), $\theta(w) \ge \theta(w^\circ)$.
\end{proof}

\begin{corollary}
\label{re:fixpoint-estim}
If $\cR_0 > 1$, $\tilde w \in M^b_+(\Omega)$,
$\tilde w = F(\tilde w)$ and $\theta(\tilde w )
>0$, then $\|\tilde w \|_\infty \ge \ln \cR_0$.
\end{corollary}
The next observation seems to be mainly the same as the one
in \cite[Thm.5.6]{BoDI1} that the basic reproduction
number associated with the situation after the epidemic outbreak is less than one.

\begin{remark}
\label{re:kernel-final-susc-intro}
Let the final sizes of the cumulative forces
of infection be described by the minimal solution
of $w = F(w) + w^\circ$.

Then  $s_\infty(x)= e^{-w(x)}$ is the  probability
that a host with trait $x$ that is susceptible at the beginning is still susceptible at the end of the epidemic,
(\ref{eq:final-size-susc-intro}).
Define the measure kernel $\kappa_\infty$ by
\[
\kappa_\infty(\omega, x) = \int_\omega s_\infty (\xi)
\kappa(d \xi, x),
\qquad
\omega \in \cB, \quad x \in \Omega.
\]
Let $\cR_\infty= \br (\kappa_\infty)$ be the spectral radius of $\kappa_\infty$ and $\theta_\infty$ be the eigenfunctional associated with $\kappa_\infty$
and $\cR_\infty$, analogously to (\ref{eq:eigenfunc-kappa}).
Then, if  $\theta_\infty(w^\circ) >0$, we have $\cR_\infty < 1$.
\end{remark}

See Remark \ref{re:kernel-final-susc} for more details.

\begin{proof}
We can assume that  $\cR_\infty > 0$,
which is the spectral radius of the operator
induced by $\kappa_\infty$ analogously to
(\ref{eq:lin-approx-intro}), and let  $\theta_\infty$ be the respective eigenfunctional of that operator associated with $\cR_\infty$
via Theorem \ref{re:Krein} and characterized by
(\ref{eq:eigenfunc-kappa}) with $\kappa_\infty$
replacing $\kappa$.
By Remark \ref{re:final-size-lin},
\[
w(x) \ge \int_\Omega w(\xi) \kappa_\infty(d\xi,x)
+ w^\circ(x).
\]
We apply the linear positive functional $\theta_\infty$ to this equation,
\[
\theta_\infty(w) \ge \cR_\infty \theta_\infty(w) + \theta_\infty(w^\circ).
\]
Since $\theta_\infty(w^\circ) >0$ by assumption,
\[
\theta_\infty(w) > \cR_\infty \theta_\infty(w),
\qquad
\theta_\infty(w) > 0,
\]
and $\cR_\infty <1 $.
\end{proof}
In the following, we explore stronger assumptions
and stronger results in the cases that the kernel
$\kappa$ is dominated or even semi-separable or that $\Omega$ is
a metric space and $\kappa$ is a tight
Feller kernel that is topologically irreducible
 \cite{Thi19}\cite[Chap.13]{Thi24}.

In particular, in Remark \ref{re:kernel-final-susc-intro},
the statement $\cR_\infty < 1$ follows under
the brute force assumption $\theta_\infty(w^\circ) >0$.
See Theorem \ref{re:eigenmeasure-intro}
and  Proposition \ref{re:theta-pos}
 for conditions under which this assumption holds,
 for instance if $\kappa$ is topologically irreducible (Definition \ref{def:Feller}) and $w^\circ $ is continuous and not the zero function.
The respective conditions are transmitted from $\kappa$
to  $\kappa_\infty$.

$\cR_\infty < 1$ may not hold if the initial
infectives have their traits in a strict subset of $\Omega$ and there is no transmission from hosts with traits in that strict
subset to hosts with traits in its complement.

%%%%%%%%%%%%%%%%%%%%%%%%%%%%%%%%%%%%%%%%%%%%%%%%%%

\subsection{Dominated measure kernels}
\label{subsec:dominated-intro}

A measure kernel $\kappa$ is called {\em dominated}
by a measure $0\ne \nu \in \cM_+(\Omega)$
if
\begin{equation}
\label{eq:kernel-dom-intro}
\kappa(\omega,x) \le \nu(\omega), \qquad \omega \in\cB, \quad x \in \Omega.
\end{equation}
The model we will consider in Sections \ref{sec:model} to \ref{sec:final-size}
 typically involves a dominated measure
kernel.

\begin{theorem}
\label{re:dominated}
Let the measure kernel $\kappa$ be dominated
by a measure $\nu$ and $\cR_0>1$.
Let $w^\circ \in M^b_+(\Omega)$
and $w$ be the  minimal solution
to $w = F(w) + w^\circ$.

Then there exists some solution $\tilde w = F(\tilde w)$ in $M^b_+(\Omega)$ with the following properties:
\begin{itemize}
\item[(a)]
 $\int_\Omega (f \circ \tilde  w) d\nu \ge \ln \cR_0$ and $\|\tilde w\|_\infty \ge \ln \cR_0$.
\item[(b)] If $\theta(w^\circ) >0$, then
 $w - w^\circ \ge \tilde  w$
 and $\ln \cR_0 \le
\int_\Omega (1 - s_\infty) d \nu.$
\end{itemize}
\end{theorem}

The last inequality follows from (a)
and the first part of (b)
and (\ref{eq:final-size-susc-intro}).

Choosing $w^\circ \equiv 1$ on $\Omega$
in the previous theorem such that $\theta(w^\circ) >0$,
we obtain the following fixed point result
(notice the scarcity of assumptions).

\begin{corollary}
\label{re:dominated-fixp}
Let the measure kernel $\kappa$ be dominated
by a measure $\nu$ and $\cR_0>1$.
Then there exists some solution $\tilde w = F(\tilde w)$ in $M^b_+(\Omega)$ with the following properties:

 $\int_\Omega (f \circ \tilde  w) d\nu \ge \ln \cR_0$ and $\|\tilde w\|_\infty \ge \ln \cR_0$.
\end{corollary}

For the proofs see Section \ref{sec:intermed-thres}.
%%%%%%%%%%%%%%%%%%%%%%%%%%%%%%%%%%%%%%%%%%%%%%%%%%%%

\subsection{Semi-separable measure kernels}
\label{subsec:semisep-intro}
%%%%%%%%%%%%%%%%%%%%%%%%%%%%%%%%%%%%%%%%

In Theorem \ref{re:dominated}, the nonzero fixed point $\tilde w$ may depend on $w^\circ$.
To remove this dependence, we consider the
following concept \cite{Ina23}.

A measure kernel $\kappa$ is called {\em semi-separable} if there are some nonzero function $k_0 \in  M^b_+(\Omega)$ and some nonzero measure
$\nu \in  \cM_+(\Omega)$ and some number $\delta \in (0,1]$ such
that
\begin{equation}
\label{eq:separ-intro}
\delta  \nu (\omega) k_0(x)\le \kappa(\omega,x)
\le
\nu(\omega) k_0(x), \qquad x \in \Omega, \quad
\omega \in \cB.
\end{equation}

Notice that the function $k_0$ is not assumed to be strictly
positive on $\Omega$ and no positivity assumption is made for $\nu$ except that it is not the zero measure. Assuming strict positivity of $k_0$
would be an easy way, though, to guarantee
that the infection reaches hosts of all traits,
in other words, the occurrence of a pandemic
in case that $\cR_0 >1$. Cf. \cite{Ina23}.
Notice that, if $\kappa$ is semi-separable,
it is dominated by the measure $\|k_0\|_\infty \,\nu$.

\begin{theorem}
\label{re:semisep-intro}
Let $\kappa$ be a semi-separable measure kernel.
Then $\cR_0 > 0$ if and only if
\[
\int_\Omega k_0 \,d \nu > 0.
\]
 Let $w $ (and $w^\circ$) be
the final sizes of the cumulative (initial)
forces of infection, i.e., $w$ is the minimal solution to $w =F(w) +w^\circ$.

\begin{itemize}
\item[(a)] If $w^\circ \in M^b_+(\Omega)$ and $\int_\Omega w^\circ
d \nu =0$, then $w = w_0$.

\item[(b)] Assume that $\cR_0 >1$. Then there exists
a unique non-zero solution $\tilde w \in  M_+^b(\Omega)$
to $\tilde w = F(\tilde w)$.

Further, $\int_\Omega (f \circ \tilde w) d \nu
\ge \ln \cR_0/\|k_0\|_\infty$ and $\|\tilde w\|_\infty \ge \ln \cR_0$.

Finally, if $\int_\Omega w^\circ \, d\nu >0$, we have $w - w^\circ \ge \tilde w$.
\end{itemize}
\end{theorem}
Notice that, in this theorem, for $\cR_0 >1$,
an epidemic outbreak occurs if and only if
$\int_\Omega w^\circ d\nu >0$ for the final size of the cumulative initial force of infection, $w^\circ$.
The next theorem explores how the threshold
property of $\cR_0$ plays out for
the final size of the epidemic in terms of the cFoI if the number
of initial infectives is very small.

\begin{theorem}
\label{re:semisep-sequ-intro}
Let $\kappa$ be a semi-separable measure kernel.
 Let  $(w_{\ell}^\circ)_{\ell \in \N}$ be a sequence
in $ M^b_+(\Omega)$
and $w_{\ell }^\circ \to 0$ as $\ell\to \infty$ pointwise
on $\Omega$.

For each $\ell \in \N$, let $w_\ell \in M^b_+(\Omega)$ be the minimal
solution of
\[
w_\ell = F(w_\ell) + w_{\ell}^\circ.
\]
\begin{itemize}
\item[(a)] If $\cR_0 \le 1$, then $w_\ell - w^\circ_\ell \to 0$
as $\ell \to \infty$, uniformly  on $\Omega$.

\item[(b)] Let $\cR_0 > 1$
and $\int_\Omega w_{\ell}^\circ \,d \nu > 0$ for all $\ell \in \N$ and
$\tilde w $ be
the unique nonzero solution to $\tilde w = F(\tilde w)$ from Theorem \ref{re:semisep-intro} (b).

Then, as $\ell \to \infty$, $w_\ell
\to \tilde w$
pointwise on $\Omega$
and $w_\ell - w^\circ_\ell \to \tilde w$
uniformly on $\Omega$.
\end{itemize}
\end{theorem}

For the proofs see Section \ref{sec:semi-sep}.
We mention that, on the way from the general
but less informative result in Theorem
\ref{re:epi-occur} and the detailed results
in this section, there are intermediate results
that are of their own interest, but will
be presented later (Section \ref{sec:intermed-thres}).

%%%%%%%%%%%%%%%%%%%%%%%%%%%%%%%%%%%%%%%%%%%%%%%%%%%%

\subsection{Metric spaces of traits and
Feller kernels}
\label{subsec:metric-intro}

We now assume that $\Omega$ is a metric space
and the $\sigma$-algebra $\cB$ is the one
generated by the open subsets of $\Omega$.
Then,
the vector space of bounded continuous real-valued functions on $\Omega$, $C^b(\Omega)$, is a closed subspace of
$M^b(\Omega)$ with the supremum norm.

\begin{definition}
[{\cite[Chap.19.3]{AlBo}}{\cite{Thi20}}{\cite[Chap.13]{Thi24}}
]
\label{def:Feller}
A measure kernel $\kappa$ is called a {\em Feller kernel} if the  operator $K$ on $M^b(\Omega)$
induced by (\ref{eq:lin-approx-intro})
maps $C^b(\Omega)$ into itself.

A Feller kernel $\kappa$ is called  {\em topologically irreducible}
if for any nonempty open strict subset $U$ of
$\Omega$ there exist some $x \in \Omega \setminus U$ such that $\kappa (U,x) > 0$.

A function $g \in M^b_+(\Omega)$ is called
{\em topologically positive} if $\inf_U g >0$
for some nonempty open subset $U$ of $\Omega$.
\end{definition}

\begin{theorem}
\label{re:irred-pos1}
Let $\kappa$ be a topologically irreducible Feller kernel and $w^\circ \in M^b_+(\Omega)$
 be topologically positive and $w \in M^b_+(\Omega)$
    be the minimal solution to $w = F(w) +w^\circ$.
Then $w(x)-w^\circ(x) > 0$ for all $x\in \Omega$.
\end{theorem}

So there is potential for a pandemic.
We mention that topological irreducibility of
$\kappa$ is necessary for this result to hold
(Remark \ref{re:irr-pos-nec}).

\subsubsection{Tight Feller kernels}
\label{subsub:tight-Feller-intro}

\begin{definition}[{\cite{Thi19}}{\cite[Chap.13]{Thi24}}]
\label{def:tight}
A measure $\mu \in \cM_+(\Omega)$ is called
{\em tight} if for any $\epsilon > 0$ there
exists a compact subset $W$ of $\Omega$
such that $\mu (\Omega \setminus W) < \epsilon$.

A Feller kernel $\kappa$ is called {\em tight}
if for any $\epsilon > 0$ there is a compact
subset $W$ of $\Omega$ such that
\begin{equation}
\label{eq:tight-intro}
\kappa(\Omega \setminus W, x) < \epsilon,
\qquad
x \in \Omega.
\end{equation}
A Feller kernel $\kappa$ is called {\em quasi-tight}
if
$\kappa = \kappa_1 + \kappa_2$ where the
$\kappa_j$ are Feller kernels, $\kappa_1$ is tight,  $\br(\kappa)
> \br(\kappa_2)$,
and for any $x \in \Omega$, there exists a separable subset $\omega $ of $\Omega$ such that $\kappa_2(x,
\Omega \setminus \omega) =0$.
\end{definition}

If $\kappa$ is a tight or quasi-tight Feller kernel, the
eigenfunctional $\theta$ in Theorem \ref{re:Krein} is given by
a measure. See \cite{Thi19}\cite[Thm.13.39]{Thi24}
and \cite[Thm.13.42]{Thi24}.

\begin{theorem}
\label{re:eigenmeasure-intro}
Let $\kappa$ be a quasi-tight Feller kernel and $\cR_0 > 0$. Then there exists a  eigenmeasure $\mu$ of
$\kappa$ associated
with $\cR_0$,
\[
\cR_0 \mu(\omega) = \int_\Omega \kappa(\omega,x) \mu(dx),
\qquad
\omega \in \cB.
\]
If, in addition, $\kappa$ is a topologically
irreducible Feller kernel, $\int_\Omega g \, d \mu >0$
for every topologically positive $g \in M^b_+(\Omega)$.
If $\kappa$ is a tight Feller kernel, $\mu$
is tight.
\end{theorem}

Now we are in the situation that the eigenfunctional $\theta$ used in Theorem \ref{re:epi-occur} (b)
is given by a measure,
\[
\theta (g) = \int_\Omega g(x) \mu (dx), \qquad g \in M^b(\Omega).
\]
In particular, $\theta $ is continuous
with respect to pointwise convergence
of increasing sequences of functions.
Theorem \ref{re:epi-occur} in conjunction with
Theorem \ref{re:eigenmeasure-intro} and
\ref{re:irred-pos1} imply the following result.

\begin{theorem}
\label{re:tight-irred-pandemic}
Let $\kappa$ be a  quasi-tight Feller kernel
 that is topologically irreducible and $\cR_0 > 1$.
 Then we have a pandemic situation:

 Let $ w^\circ \in  M^b_+(\Omega)$ be topologically
 positive and $w \in M^b_+(\Omega)$ be the minimal
 solution of  the final
size equation $w = F(w) +w^\circ$.

Then,  $w- w^\circ$ is strictly positive on $\Omega$
and  $\|w- w^\circ\|_\infty \ge \ln \cR_0$ and
$\displaystyle \inf_{\Omega} s_\infty \le 1/\cR_0$.
\end{theorem}

\subsubsection{Strong Feller kernels}
\label{subsub:strong-Feller-intro}

Again, the last result shows the threshold property of $\cR_0$,
though not in a very strong sense. To compare
the final size solutions to nonzero fixed points of $\tilde w =F(\tilde w)$, we strengthen
 the Feller property.

\begin{definition}
\label{def:strong-Fel}
A measure kernel $\kappa$ is called a {\em strong
Feller kernel} if $\kappa(\omega,\cdot)$ is
continuous on $\Omega$ for all $\omega \in \cB$.
\end{definition}

A dominated Feller kernel (Section \ref{subsec:dominated-intro})
is a strong Feller kernel (Proposition \ref{re:dominated-Feller-strong}).
See Example \ref{exp:strong-Feller} for a measure
kernel that under weak assumptions
is a Feller kernel though not a strong Feller
kernel, but  is a strong Feller kernel
under stronger assumptions.
As we will show (Lemma \ref{re:Feller-strong}), the maps $K$ and $F$ associated with a strong Feller kernel map $M^b_+(\Omega)$ into
$C^b_+(\Omega)$.

\begin{theorem}
\label{re:Feller-strong-intro}
 Let $\kappa$ be a strong Feller kernel
 that is tight and topologically irreducible and let $\cR_0 > 1$.
  Let $w^\circ \in M^b_+(\Omega)$ be
  topologically
 positive.

 Let $w \in M^b_+(\Omega)$
 be the  minimal solution to
 $w = F(w) + w^\circ$.

 Then there exists some $\tilde w \in C^b_+(\Omega)$ such that
   $\tilde w$ strictly positive on $\Omega$, $\| \tilde w \|_\infty \ge \ln \cR_0$
 and $w -w^\circ \ge \tilde w  = F(\tilde w)$.
\end{theorem}

For the proof of this and related results see
Theorem \ref{re:tight-Feller-converge}
and its corollaries.

%%%%%%%%%%%%%%%%%%%%%%%%%%%%%%%%%%%%%%%%%%%%%%%%%%%

\subsubsection{Comparability kernels}
\label{subsub:compare-intro}

In the last result, the fixed point $\tilde w$
depends on $w^\circ$. To enforce independence,
we introduce the following concept (Section
\ref{subsec:compar-ker}). Cf. \cite[(K-2)]{Thi79}.

\begin{definition}
\label{def:compar-ker}
A Feller kernel $\kappa$ is called a
{\em comparability kernel} if for any
continuous function $g:\Omega \to (0,\infty)$
there exists some $\delta > 0$ such that
\begin{equation}
\label{eq:comparabilator-intro}
 \int_\Omega g(\xi) \kappa(d\xi,x) \ge \delta \, \kappa(\Omega, x), \qquad x \in \Omega.
\end{equation}
\end{definition}

Sufficient conditions and examples will be presented in Section \ref{subsubsec:compar-ker-exp}.

The next two theorems explores how the threshold
property of $\cR_0$ plays out for
the final size of the epidemic in terms of the cFoI if the number
of initial infectives is very small. The results parallel those of
Theorem \ref{re:semisep-sequ-intro}, but even if $\Omega$ is
a metric space the assumptions are different.
See the end of the discussion in Section \ref{subsec:disc-intro}.

\begin{theorem}
\label{re:ultimate-intro}
Let
$\kappa$ be a  strong Feller kernel that is
topologically irreducible and tight.
Assume that $\kappa$ is a comparability kernel and $\cR_0>1$.

\begin{itemize}

\item[(a)] Then there exists a unique nonzero solution     $\tilde w \in  C^b_+(\Omega)$ to
the equation $\tilde w = F(\tilde w)$; further $ \|\tilde w \| \ge \ln \cR_0$.

\item[(b)] If   $(w^\circ_{\ell})_{\ell \in \N}$ is a decreasing sequence of topologically positive functions
in $ M^b_+(\Omega)$
and $w^\circ_{\ell} \to 0$ as $\ell \to \infty$ uniformly on all compact subsets
of $\Omega$, then $w_\ell - w^\circ_\ell \to \tilde w$ uniformly on  $\Omega$
for the minimal solutions $w_\ell$ of $w_\ell = F(w_\ell) + w^\circ_{\ell}$.
\end{itemize}
\end{theorem}

\begin{theorem}
\label{re:epino-intro}
Let
$\kappa$ be a   Feller kernel that is
topologically irreducible.
Assume that $\kappa$ is a comparability kernel and $\cR_0 \le 1$.

\begin{itemize}
\item[(a)] Then there exists no solution $\tilde w \in \dot C^b_+(\Omega)$ to
the equation $\tilde w = F(\tilde w)$.

\item[(b)] If $\kappa$ is a strong Feller kernel and $(w^\circ_{\ell})_{\ell \in \N}$ is a   sequence
in $ M^b_+(\Omega)$
and $w^\circ_{\ell} \to 0$ as $\ell \to \infty$
pointwise
on $\Omega$, then $w_\ell  \to 0$ as $\ell \to \infty$ pointwise
for the minimal solutions $w_\ell$ of $w_\ell = F(w_\ell) + w^\circ_{\ell}$.
\end{itemize}
\end{theorem}

For the proofs of these results see Section
\ref{subsec:compar-ker}.
%%%%%%%%%%%%%%%%%%%%%%%%%%%%%%%%%%%%%%%%%%%%%%%%%%%%%%%

\section{Discussion}
\label{subsec:disc-intro}

%%%%%%%%%%%%%%%%%%%%%%%%%%%%%%%%%%%%%%%%%%%%%%%%%%%%%%

The results in this paper are based on
finding a positive eigenfunctional of a bounded linear
positive operator associated with its spectral
radius (Theorem \ref{re:Krein}). The existing literature mostly attempts
to find a positive eigenvector in addition, \cite{BoDI1, Ina23}, e.g. In the context of measure kernels, one
can take that route, too, if one imposes
a uniform Feller property.

A measure kernel $\kappa$ is called a {\em uniform
Feller kernel} \cite[Sec.13.6]{Thi24} if
\begin{equation}
\label{eq:uniform-Feller}
\sup_{\omega \in \cB}  \big | \kappa(\omega,x)  -  \kappa(\omega, x_0)  \big |
\to 0, \qquad x \to x_0 \in \Omega.
\end{equation}
If $\kappa$ is given in the form
\begin{equation}
\label{eq:meas-ker-susc-intro}
\kappa(\omega,x) = \int_\omega k(x,\xi)S_0(d\xi), \qquad \omega \in \cB,
\quad x \in \Omega,
\end{equation}
as it is from
 Section \ref{sec:model} to Section \ref{sec:final-size} following \cite{BoDI1}, (\ref{eq:uniform-Feller}) is
equivalent to
\[
\int_\Omega \big | k(x,\xi) - k(x_0,\xi)\big |\,S_0(d\xi)
\longrightarrow 0, \qquad x \to x_0.
\]
Cf. hypothesis  ${\rm H}_{A_1}$ in \cite[Sec.5]{BoDI1}. See also assumption (41) in \cite{Ina23}.

General measure kernels appear in
the final size equation for the epidemic
model in \cite[Sec.3]{Thi77}.
If $\kappa$ is a uniform Feller kernel, existence of a (strictly positive) eigenvector
of $K$ follows from \cite{Thi20}, \cite[Thm.13.52]{Thi24} and \cite[Thm.13.58]{Thi24}.

\begin{theorem}
\label{re:eigenvector}
Let $\kappa$ be a tight uniform Feller kernel
and $\cR_0= \br(\kappa) >0$. Then there exists an eigenvector
$ v \in  C^b_+(\Omega)$ such that
$\cR_0 v = K v$. If $\kappa$ is  topologically
irreducible in addition, $v$ is strictly positive on $\Omega$.
\end{theorem}

A somewhat more general result is
proved in \cite{Thi20}\cite[Thm.13.58]{Thi24}.
If $\kappa$ is given by (\ref{eq:meas-ker-susc-intro}),
 we could neither find
a strong  Feller kernel that is not
a uniform Feller kernel  nor could we prove
that any strong Feller kernel is a uniform
Feller kernel.
If we do not insist on an epidemiologic
interpretation, we can consider the following
measure kernel.
\begin{example}
\label{exp:strong-Feller}
Let $\Omega = \R$ with the standard topology
induced by the absolute value. Consider
\[
\kappa (\omega,x) = \int_\R \chi_\omega \big(\xi \phi(x)\big ) k_0(\xi)\, d\xi,
\qquad
\omega \in \cB, \quad x \in \R,
\]
where $k_0 : \R \to \R_+$ is finitely integrable
and $\phi:\R \to (0,\infty)$ is continuous.

Here, $\chi_\omega$ is the characteristic
or indicator function of the set $\omega$,
$\chi_\omega (x) =1$ if $x\in \omega$
and $\chi_\omega(x) = 0$ if $x \in \Omega \setminus \omega$.
By Lebesgue's theorem of dominated convergence,
$\kappa$ is a Feller kernel. After a substitution,
\[
\kappa (\omega,x) = \int_\omega   k_0 \big(\eta/\phi(x)\big) \,(1/\phi(x)) \,d\eta,
\qquad
\omega \in \cB, \quad x \in \R.
\]
Assume that $k_0$ is lower semicontinuous.
Let $x_n \to x$ as $n \to \infty$.
By Fatou's lemma,
\[
\liminf_{n \to \infty} \kappa (\omega, x_n)
\ge
\int_\omega \liminf_{n \to \infty} k_0\big (\eta/\phi(x_n)\big )\,(1/\phi(x_n))
\, d \eta.
\]
Since $k_0$ is lower semicontinuous and $\phi$ is continuous,
\[
\liminf_{n \to \infty}
k_0 \big (\eta/\phi(x_n)\big)\,(1/\phi(x_n)) \ge k_0\big(\eta/\phi(x)\big)(1 /\phi(x)), \qquad \eta \in \R.
\]
See  \cite[L.A.49]{Thi24} which also implies that
$\kappa(\omega, \cdot)$ is lower semicontinuous.
By Proposition \ref{re:kappa-lower-sc},
$\kappa$ is a strong Feller kernel.

$\kappa$ is a uniform Feller kernel if and only if,
for any sequence $x_n \to x \;(n \to \infty)$,
\[
\int_\Omega  \Big | k_0\big(\eta/\phi(x_n)\big) \, (1/\phi(x_n)) -
k_0\big (\eta/\phi(x)\big)\, (1/\phi(x)) \Big | \,d\eta \to 0,
\qquad n \to \infty.
\]
We cannot present a lower semicontinuous
$k_0 \in L^1(\R)$ that does not satisfy this
condition, but if one would like to come
up with  sufficient conditions they would presumably
 entail Lebesgue's theorem of dominated
 convergence involving a.e. continuity of $k_0 $  and some  domination condition stronger than just $k_0$ being finitely  integrable.
\end{example}

Our epidemic model considers density-dependent
alias mass action incidence as do the models
in \cite{BoDI1, Ina23}. See (\ref{eq:incidence-intro}).
More generally, we could consider
\begin{equation}
\label{eq:porb-sus-diff1}
\p_t s(t,x) = - \phi(s(t,x)) I(t,x),
\qquad s(0,x) =1.
\end{equation}
Here $\phi:\R_+ \to \R_+$ is continuous and strictly
increasing,
positive on $(0,\infty)$,  $\phi(0)=0$, $\phi(1) =1$ and $\int_0^1 (1/\phi(r)) dr = \infty$. Then
\[
1- s(t,x) = f (J(t,x)),
\]
where $f$
is an increasing concave function with strictly
decreasing derivative, $f(0)=0$, $f'(0)=1$. Cf. \cite[(5)(6)]{Thi77}.

Then the result $\|w -w^\circ\|_\infty \ge \ln \cR_0>0$
in Theorem \ref{re:epi-occur} and other theorems
has the more general form $f'(\|w -w^\circ\|_\infty) \le 1/\cR_0 < 1$.
See the proof of Theorem \ref{re:epi-occur} (b).

The generalization
\begin{equation}
\label{eq:porb-sus-diff2}
\p_t s(t,x) = - s(t,x) \, \psi( I(t,x)),
\qquad s(0,x) =1,
\end{equation}
with a continuous increasing $\psi: \R_+ \to \R_+$
seems to be much more difficult to handle. For piecewise linear $\psi$ see \cite{LuSt-free}, and for
general incidences (and no host population structure) see \cite{FCGTage}.

There are two different though not disjoint sets of assumptions under which
the threshold properties of $\cR_0$ for epidemics with small numbers of initial infectives can be rather completely
described: if $\kappa$ is a semi-separable measure kernel (Theorem
\ref{re:semisep-sequ-intro})
or if $\Omega $ is a metric space and $\kappa$
is a strong Feller kernel that is tight and a topologically irreducible comparability kernel (Theorem \ref{re:ultimate-intro} and \ref{re:epino-intro}).

On the one hand, every semi-separable measure kernel is
a comparability kernel (Proposition \ref{re:separ-envelop}).

On the other hand, a semi-separable Feller kernel is tight if and only if the measure $\nu$ in (\ref{eq:separ-intro}) is tight. Moreover, a semi-separable Feller kernel is topologically
irreducible (Definition \ref{def:Feller}) only if $\nu(U) > 0$ for the measure $\nu$ in (\ref{eq:separ-intro}) and for every non-empty open
subset $U$ of $\Omega$.

This paper does not deal with the important numerical
computation of $\cR_0$. Alternative characterizations
of the spectral radius, which one could use as a
starting point for this task,
can be found in \cite[Sec.5.5, Chap.6]{Thi24}; an iterative  approximation of the spectral radius  is proposed in \cite[Sec.6.4]{Thi24}. The numerical procedures
we are aware of, for example \cite{BrRiVe, DeSeVe} and the
references therein,  seem to
rely on the Krein-Rutman theorem for compact operators, which are approximated by operators the spectral
radii of which can be determined more easily.
It is one of the points of this paper   to avoid compactness assumptions and
use Theorem \ref{re:Krein} for cones with interior points instead; this result can  be
found in \cite{KrRu} as well. Without compactness,
continuous dependence of the spectral radius
on its operator may become a problem \cite{LeNu,
Thiconti} \cite[Ch.10]{Thi24}.

%%%%%%%%%%%%%%%%%%%%%%%%%%%%%%%%%%%%%%%%%%%%%%%%%%%%

%%%%%%%%%%%%%%%%%%%%%%%%%%%%%%%%%%%%%%%%%%%%%%%%%%

\section{The epidemic model}
\label{sec:model}

We consider a host population which is structured
by characteristic traits $x \in \Omega$, where
$\Omega$ is a measurable space with $\sigma$-algebra $\cB$. The model starts at time $t =0$. Our derivation follows
\cite{BoDI1}, but we   consider an epidemic
which is triggered by initial infectives.
%%%%%%%%%%%%%%%%%%%%%%%%%%%%%%%%%%%%%%%%%%%%%%

\subsection{Susceptible hosts, incidence and force of infection}

Let $I(t,x)$ be the force of infection that
affects susceptible hosts with trait $x \in \Omega$ at time $t \ge 0$. Let us consider a typical host with trait $x$
and $s(t,x)$ be the probability that this host
is still susceptible at time $t$ provided it was
susceptible at time $0$ \cite{BoDI1}. Then
\begin{equation}
\label{eq:prob-susc-FoI}
s(t,x) = e^{-J(t,x)}, \qquad J(t,x) = \int_0^t I(r,x)dr, \qquad t \ge 0, \quad x \in \Omega,
\end{equation}
provided that the integral makes sense. $J(t,x)$ is the cumulative force of infection affecting susceptible hosts
with trait $x$, accumulated from time 0 to time $t$. $J$ is certainly well-defined if $I(\cdot,x)$
is continuous and, by the fundamental theorem
of calculus and the chain rule,
\begin{equation}
\label{eq:porb-sus-diff}
\p_t s(t,x) = - s(t,x) I(t,x).
\end{equation}
$J$ is also well-defined if $I(\cdot,x)$
is finitely integrable on bounded subintervals
of $\R_+$. Then $s(\cdot,x)$ is absolutely
continuous and the last differential equation
holds for a.a. $t \ge 0$ \cite[L.8.1]{PoTh}.
In other words,
\begin{equation}
\label{eq:prob-non-inf}
1 - s(t,x) = \int_0^t s(r,x) I(r,x) dr,
\qquad
t \ge 0, \quad x\in \Omega.
\end{equation}

Following \cite{BoDI1, BoDI2}, we assume that
 the structural distribution of susceptible
hosts at time 0 is described by a non-negative measure $S_0: \cB \to \R_+$. Since the size of the susceptible
population is finite, $S_0(\Omega) < \infty$.  The structural distribution
of susceptible hosts at time $t\ge 0$ is given by
\begin{equation}
\label{eq:suscept}
S(t,\omega) = \int_\omega s(t,x) S_0(dx),
\end{equation}
where $S(t,\omega)$ is the number of susceptible
hosts with trait in the set $\omega \in \cB$ at time $t \ge 0$.

%%%%%%%%%%%%%%%%%%%%%%%%%%%%%%%%%%%%%%%%%%%%%%%%

\subsection{Incidence}

We redo the modeling steps because differently
from \cite{BoDI1} we do not consider an
epidemic with prehistory but an emerging epidemic,
for instance if a pathogen
overcomes a species barrier.

$S_0(\omega) - S(t,\omega)$ is the number
of  infected host with trait in $\omega \in \cB$ at time $t$.
By (\ref{eq:suscept}),
 (\ref{eq:prob-susc-FoI}), (\ref{eq:prob-non-inf}) and Tonelli's theorem,
changing the order of integration,
\begin{equation}
\label{eq:incid-suscep}
S_0(\omega) - S(t,\omega)
= \int_\Omega (1 -s(t,x)) S_0(dx) = \int_0^t B(r,\omega) dr,
\end{equation}
\begin{equation}
\label{eq:incid}
B(t, \omega) = \int_\omega s(t,x) I(t,x) S_0(dx),
\qquad t \ge 0, \quad \omega \in \cB.
\end{equation}
$B(t,\omega)$ is the incidence of the disease,  the
rate of new infections with trait in $\omega$ at time $t$.
So to speak, $B$ is the birth rate of the disease.

%%%%%%%%%%%%%%%%%%%%%%%%%%%%%%%%%%%%%%%%%%%%%%

\subsection{Force of infection and infection age}

The infection age $a$ of an infected host is the
time since its infection.
Let $u(t, \cdot,\omega)$ be the infection-age
density of infected hosts at time $t$ and trait
in $\omega \in \cB$. Recall the incidence $B$, the rate
of infections,
\begin{equation}
\label{eq:inf-age}
u(t,a , \omega) = \int_\omega P(a,x) B(t-a, dx),
\qquad t> a, \quad \omega \in \cB.
\end{equation}
Here, $P(a,x) \in [0,1]$ is the probability that a host
with trait $x$ that has been infected  in the past is still infected (in particular alive) at infection age $a$, $P(0,x) =1$ and
$P(\cdot,x)$ decreasing on $\R_+$ for each $x \in \Omega$.
By (\ref{eq:incid}),
\begin{equation}
\label{eq:inf-age-later}
u(t,a,\omega) = \int_\omega P(a,x)
 s(t-a,x) I(t-a,x) S_0(dx), \qquad t > a, \omega \in \cB.
\end{equation}
Let $u_0 (a, \omega)$ be the hosts with trait
in the set $\omega \in \cB$ and infection age $a$ at time 0.
Then
\begin{equation}
\label{eq:inf-age-init}
u(t,a, \omega) = \int_\omega Q(a,a-t,x) u_0(a-t, dx),
\qquad a > t, \quad \omega \in \cB.
\end{equation}
Here $Q(a,s,x)\in [0,1]$ is the probability that an
infected host with trait $x$ and  with infection age $s$ at the beginning is still alive and infected at age $a >s$.
Often, one chooses
\begin{equation}
\label{eq:prob-remain-inf-life}
Q(a,s ,x) = \frac{P(a,x)}{P(s,x)}, \qquad a > s >0. \end{equation}
However,  this assumes that the infection of the initially infected hosts occurred in the same way as after the
start of the epidemic.

The force of infection  that affects a typical susceptible individual of trait $x$ at time $t$
is given by
\begin{equation}
\label{eq:FoI}
I(t,x) = \int_0^\infty \Big (\int_\Omega \eta(x,a, \xi)
u(t, a, d \xi) \Big ) da , \qquad t \ge 0, \quad
x \in \Omega,
\end{equation}
where $\eta(x,a,\xi)$ indicates how a susceptible
host with trait $x \in \Omega$ is affected by an infected host with trait $\xi \in \Omega$ and infection age $a$.
By  (\ref{eq:inf-age-later}) and (\ref{eq:inf-age-init}),
\begin{equation}
\label{eq:FoI2}
I(t,x) = \int_0^t \Big ( \int_\Omega A(x,a,\xi)
 s(t-a,\xi) I(t-a,\xi) S_0(d \xi) \Big ) da
+ I_0(t,x),
\end{equation}
\begin{equation}
\label{eq:infectivity-combined}
A(x,a,\xi) = \eta(x,a,\xi)
P(a, \xi), \qquad x,\xi \in \Omega, \quad a \ge 0,
\end{equation}
and $I_0$ is the force of infection due
to the initially infected hosts,
\begin{equation}
\label{eq:FoIinit}
I_0(t,x) = \int_0^\infty \Big (\int_\Omega \eta(x,t+a,\xi)
Q(t+a,a, \xi) u_0(a,d\xi) \Big ) da.
\end{equation}

%%%%%%%%%%%%%%%%%%%%%%%%%%%%%%%%%%%%%%%%%%%%%%%%%%

\subsection{The cumulative  force of infection}
\label{subsec:model equation}

Let
\begin{equation}
\label{eq:cum-FoI-initial}
J_0(t,x) = \int_0^t I_0(r,x) dr, \qquad t \ge 0, \quad x \in \Omega,
\end{equation}
be the cumulative force of infection due to
the initially infected hosts.
By Tonelli's theorem and (\ref{eq:prob-susc-FoI})
and (\ref{eq:FoI2}),
\[
\begin{split}
&J(t,x) - J_0(t,x)= \int_0^t \big (I(r,x) dr - I_0(r,x) \big ) dr
\\
= &
\int_\Omega S_0(d \xi) \int_0^t   A(x,a,\xi)
\Big ( \int_a^t  s(r-a,\xi) I(r-a,\xi)  dr \Big ) da.
\end{split}
\]
After a change of variables and by (\ref{eq:prob-non-inf}),
\[
\begin{split}
J(t,x) - J_0(t,x)= &
\int_\Omega S_0(d \xi) \int_0^t   A(x,a,\xi)
\Big ( \int_0^{t-a}  s(r,\xi) I(r,\xi)  dr \Big ) da .
\\ = &
\int_\Omega S_0(d \xi) \int_0^t   A(x,a,\xi)
\Big ( 1- s(t-a,\xi)    \Big ) da .
\end{split}
\]
By (\ref{eq:prob-susc-FoI}),
\begin{equation}
\label{eq:cum-inf-force-eq}
J(t,x) - J_0(t,x)
=
\int_\Omega  S_0(d \xi) \int_0^t   A(x,a,\xi)
f \big (J(t-a,\xi)    \big ) da,
\end{equation}
with
$f(J) = 1 - e^{-J}$, $J \in \R$, (\ref{eq:skellam-intro}). Cf. \cite[(10)]{Ina23}.
By (\ref{eq:FoIinit}) and (\ref{eq:cum-FoI-initial}),
\begin{equation}
\label{eq:cum-FoI-initial-explic}
J_0(t,x) = \int_0^\infty \int_\Omega \Big( \int_0^t \eta(x,r+a,\xi) Q(r+a,a, \xi)dr\Big ) u_0(a,d\xi)  da.
\end{equation}

\begin{assumption}
\label{ass:initial-inf-bounded}

The functions $\eta, P, Q $ are nonnegative
and are measurable on their respective domains
equipped with the appropriate product
$\sigma$-algebras. $u_0: \R_+ \times \cB \to \R_+$ has the analogous
properties of a measure kernel,
\begin{equation}
\label{eq:init-infec-age}
\int_0^\infty u_0 (a, \Omega) da < \infty.
\end{equation}
The integrals
\begin{equation}
\label{eq:init-influ-cum}
\int_0^\infty \eta(x,r+a, \xi) Q(r+a,a, \xi) dr
\end{equation}
 provide a bounded function of $(x,\xi,a) \in \Omega^2 \times \R_+$. The integrals
\begin{equation}
\int_\Omega S_0(d\xi) \int_0^\infty
A(x,a,\xi) da
\end{equation}
with $A$ from (\ref{eq:infectivity-combined})
provide a bounded function of $x \in \Omega$.

\end{assumption}

\begin{remark} Assumption (\ref{eq:init-influ-cum})
holds, e.g., if $\eta:\Omega \times \R_+ \times
\Omega \to \R_+$ is measurable and bounded and $D:\R_+ \times \Omega \to \R_+$ given by
\begin{equation}
D(a,\xi) = \int_0^\infty Q(r+a,a,\xi) dr,
\qquad a \in \R_+, x \in \Omega,
\end{equation}
is a bounded function. If $Q$ is given by
(\ref{eq:prob-remain-inf-life}), $D(a,\xi)$
is the expected duration of remaining infected
life at infection-age $a$ with trait $\xi$.
See \cite[Sec.12.4]{Thi03}.
\end{remark}

%%%%%%%%%%%%%%%%%%%%%%%%%%%%%%%%%%%%%%%%%%%%

\section{Existence of minimal solutions}
\label{sec:existence}

%%%%%%%%%%%%%%%%%%%%%%%%%%%%%%%%%%%%%%%%%%%%%

As one can expect from (\ref{eq:cum-FoI-initial},  $J_0(\cdot, x)$ is increasing
on $\R_+$ for all $x \in \Omega$. See
(\ref{eq:cum-FoI-initial-explic}).
Guided by (\ref{eq:cum-inf-force-eq}), for $n \in \Z_+$, we define inductively
\begin{equation}
\label{eq:recursive}
J_{n+1}(t,x) = J_0(t,x)
+
\int_\Omega S_0(d \xi) \int_0^t   A(x,a,\xi)
  f(J_n(t-a,\xi)) da .
\end{equation}
By induction, since $f$ in (\ref{eq:skellam-intro}), is increasing,
$
J_{n+1} (t,x) \ge J_n(t,x)$, $ t \in \R_+, x \in \Omega,
$
and
$J_n(\cdot, x)$ is increasing on $\R_+$ for all
$n \in \Z_+$, $x \in \Omega$. Further, for $n \in \N$,
\begin{equation}
J_n(t,x) \le J_0(t,x) + \int_\Omega S_0(d\xi) \int_0^t
A(x,a,\xi) da ,
\qquad t \in \R_+, x \in \Omega.
\end{equation}
For all $(t,x)$, the sequences $(J_n(t,x))$
are increasing and bounded and the limits
\begin{equation}
J(t,x) := \lim_{n \to \infty} J_n(t,x), \qquad
t \in \R_+, \quad x \in \Omega,
\end{equation}
exist pointwise.
By Beppo Levi's theorem of monotone convergence,
we can take the limit $n \to \infty$ in
(\ref{eq:recursive}) and obtain that $J$ is a
solution of (\ref{eq:cum-inf-force-eq})
and $J(\cdot,x)$ is increasing for all $x \in \Omega$.

Let $\tilde J$ be also a solution  of (\ref{eq:cum-inf-force-eq}).
By induction, $J_n(t,x) \le \tilde J(t,x)$ for all $n \in \N$
and  $J(t,x) \le \tilde J(t,x)$.
So $J$ is the minimal solution of (\ref{eq:cum-inf-force-eq}).
Of course, there is at most one minimal solution.
 $J_n$ can be interpreted as the
cumulative infective force due to the infected
generations from the initial ($0^{th}$) to the $n^{th}$
 generation. This suggests that the minimal solution,
 which is their limit, is the epidemiologically
 relevant solution. In summary:

 \begin{theorem}
 \label{re:minimal-sol-Vol-Ham}
 There is a minimal solution $J$ of (\ref{eq:cum-inf-force-eq}) which is the
 monotone limit
 of the recursion (\ref{eq:recursive}).
 $J$ is the epidemiologically relevant solution
 of (\ref{eq:cum-inf-force-eq}).
 \end{theorem}

Analogously, as in \cite{Thi77}, we could
derive conditions for the minimal solution
to be the only solution  of (\ref{eq:cum-inf-force-eq}).

\subsection{The case of a positive minimum latency period}

Assume that there is some $a_0 > 0$ such that
\begin{equation}
\eta(x, a,\xi) =0 , \qquad a \in [0,a_0), \quad
x, \xi \in \Omega.
\end{equation}
Then, for all traits $x$, there is a latency
period with length greater or equal to $a_0$.

\begin{theorem}
\label{re:min-lat-per0}
For $t \le na_0$, $J(t,x) = J_n(t,x)$
for all $x \in \Omega$.
\end{theorem}

\begin{proof}
This holds for $n=1$. Let $n \in \N$ and the assertion is true for $n$. By (\ref{eq:recursive}), $J(t,x) = J_{n+1}(t,x)$ for $ t \in [0,na_0]$.
Let $t \in (na_0,(n+1) a_0)$.
By (\ref{eq:cum-inf-force-eq}),
\[
J(t,x) - J_0(t,x) = \int_\Omega S_0(dx) \int_{a_0}^t A(x,a,\xi) f(J(t-a,x) ) da.
\]
Since $J(t-a,x ) = J_n(t-a,x)$ for $t \le (n+1)a$
and $a \ge a_0$, $J(t,x) = J_{n+1} (t,x)$ for
$t \le (n+1)a$, $x \in \Omega$.
\end{proof}

%%%%%%%%%%%%%%%%%%%%%%%%%%%%%%%%%%%%%%%%%%%%%%%%%%%%%

\section{The final size of the epidemic}
\label{sec:final-size}

%%%%%%%%%%%%%%%%%%%%%%%%%%%%%%%%%%%%%%%%%%%%%%%%%%

Recall that the minimal solution $J$ of (\ref{eq:cum-inf-force-eq})
has the property that $J(t, x)$
is an increasing bounded function of $t \ge 0$
for any $x \in \Omega$. Hence, the limit
\begin{equation}
\label{eq:final-cum-FoI}
w(x)   = \lim_{t\to \infty} J(t,x), \quad x \in \Omega,
\end{equation}
exists for all $x \in \Omega$. By (\ref{eq:cum-FoI-initial-explic}),
\begin{equation}
\label{eq:final-cum-init-FoI}
\begin{split}
J_0(t,x) \nearrow w_0(x)
:= &
\int_0^\infty \int_\Omega A_0(x,a, \xi) u_0(a,d\xi)  da,
\\
A_0(x,a,\xi) = &  \int_0^\infty \eta(x,r+a,\xi) Q(r+a,a, \xi)dr.
\end{split}
\end{equation}
By Beppo Levi's theorem of  monotone convergence (or the Lebesgue-Fatou lemma
\cite[p.468]{Thi03}),
we can take the limit $t \to \infty$ in
(\ref{eq:cum-inf-force-eq})  (cf. \cite[(31)]{Ina23}),
\begin{equation}
\label{eq:cum-inf-force-final}
w(x) =w_0(x)
+
\int_\Omega  k(x,\xi)
  f( w(\xi))  S_0(d \xi), \qquad x \in \Omega,
\end{equation}
with $f$ in (\ref{eq:skellam-intro}) and
\begin{equation}
\label{eq:integral-kernel}
k(x,\xi) = \int_0^\infty A(x,a,\xi) da, \qquad x, \xi \in \Omega.
\end{equation}

\begin{theorem}
The function $w: \Omega \to \R_+$ given by
\begin{equation}
\label{eq:final-size}
w(x) = \lim_{t\to \infty} J(t,x),
\qquad x \in \Omega,
\end{equation}
is the minimal solution of
(\ref{eq:cum-inf-force-final}) and is obtained as
pointwise limit
\begin{equation}
\label{eq:final-limit}
w(x) = \lim_{n\to \infty} w_n(x), \quad x \in \Omega,
\end{equation}
 of the recursion
\begin{equation}
\label{eq:final-recursive}
w_{n+1}(x) = \int_\Omega k(x, \xi) f( w_n(\xi))  S_0(d \xi) +w_0(x), \qquad n \in \Z_+, \quad x \in \Omega.
\end{equation}
Further, for all $x\in \Omega$ and $n \in \N$
\begin{equation}
\label{eq:time-limit-gener}
J_n(t,x) \nearrow w_n(x), \qquad  t \nearrow \infty.
\end{equation}
\end{theorem}
We mention that the concept of a minimal
solution for the final size equation has
already been considered in \cite[Thm.3.2]{Thi77}.
By (\ref{eq:final-size}) and (\ref{eq:prob-susc-FoI}),
\begin{equation}
\label{eq:cFoI-susc-final}
e^{-w(x)} = \lim_{t \to \infty} s(t,x) =: s_\infty(x),
\qquad
x \in \Omega.
\end{equation}

\begin{proof}
Since $f$ is increasing, by induction, for each $n \in \N$, $w_{n+1}(x)
\ge w_n(x)$ for all $x \in \Omega$. Further,
$\{w_n(x); x \in \Omega, n \in \Z_+\}$ is bounded.
So, the pointwise limit $\tilde w(x) = \lim_{n \to
\infty} w_n(x)$ exists and, as in the proof of
Theorem \ref{re:minimal-sol-Vol-Ham}, $\tilde w$ is the minimal solution of
(\ref{eq:cum-inf-force-final}). Recall  $J(t,x)=
\lim_{n \to \infty} J_n (t,x)$ and the recursion
(\ref{eq:recursive}). By induction,
\[
J_n(t,x) \le w_n(x), \qquad n \in \N, \quad t \ge 0, \quad x \in \Omega.
\]
So $J(t,x) \le \tilde w(x)$ and, by (\ref{eq:final-size}), $w(x) \le
\tilde w(x)$ for all $x \in \Omega$. Since $w$
is a solution of (\ref{eq:cum-inf-force-final})
and $\tilde w$ is the minimal solution of
(\ref{eq:cum-inf-force-final}), we also have $\tilde w(x) \le w(x)$ for all $x \in \Omega$.
The monotone convergence in (\ref{eq:time-limit-gener}) follows inductively from
(\ref{eq:recursive}) and Beppo Levi's theorem
of monotone convergence
by taking the limit for $t \to \infty$
and by using (\ref{eq:final-recursive}).
\end{proof}

Using the concavity of the function $f$,
one can derive conditions which make the minimal
solution the only solution of (\ref{eq:cum-inf-force-final}). But since the
minimal solution is the epidemiologically relevant
solution, we will not go into the technicalities
of such a proof yet \cite{Thi79}. See Theorem
\ref{re:unique}.

\begin{proposition}
\label{re:pos-min-lat}
Assume that $\eta(x, a, \xi) = 0$ for all $a\in
[0,a_0]$, $x \in \Omega$.
Then $J(t,x) \le w_n(x) $ for $t \ge na_0$, $x \in \Omega$ where $J$ is the minimal solution
of (\ref{eq:cum-inf-force-eq}).
\end{proposition}

\begin{proof}
This follows from Theorem \ref{re:min-lat-per0}
and (\ref{eq:time-limit-gener}).
\end{proof}

%%%%%%%%%%%%%%%%%%%%%%%%%%%%%%%%%%%%%%%%%%%%%%%%%%%

\subsection{Measure kernels}
\label{subsec:measure-kernels}

Recall the concept of a measure kernel, (\ref{eq:kernel-meas-intro}).
Define
\begin{equation}
\label{eq:meas-ker-susc}
\kappa(\omega, x) = \int_\omega k(x,\xi) S_0(d\xi),
\qquad x \in \Omega, \quad \omega \in \cB,
\end{equation}
with $k$ from (\ref{eq:integral-kernel}).
To make $\kappa$ a measure kernel, (\ref{eq:kernel-meas-intro}),
we assume that
\[
\int_\Omega k(\cdot ,\xi) S_0(d\xi) \hbox{ is bounded on } \Omega.
\]
Notice that, with this $\kappa$, the operator $K$ on $M^b(\Omega)$
induced by (\ref{eq:lin-approx-intro}) corresponds
to the operator $\cK_0$ in \cite[(3.7)]{BoDI1}
which is there interpreted as an NGO, represented
in fractions. The measure kernel $\kappa$ in
(\ref{eq:meas-ker-susc}) does not seem to have a tangible  epidemiological interpretation.

\begin{remark}
\label{re:kernel-final-susc}
Recall the measure kernel
\[
\kappa_\infty(\omega,x) = \int_\omega s_\infty (\xi)
\kappa(d \xi,x), \qquad \omega \in \cB, \quad x \in \Omega,
\]
in Remark \ref{re:kernel-final-susc-intro}.
By (\ref{eq:meas-ker-susc}),
\[
\kappa_\infty (\omega,x) = \int_\omega k(x,\xi) s_\infty(\xi) S_0(d \xi), \qquad \omega \in \cB,
\quad x \in \Omega,
\]
is the measure kernel associated with the final trait distribution of the susceptible hosts,
\[
S_\infty(\omega) = \int_\omega s_\infty(\xi) S_0(d \xi), \qquad \omega \in \cB.
\]
\end{remark}
Recall the function $f$ in  (\ref{eq:skellam-intro}),
$f(r) =1 -e^{-r}$, $ r \in \R_+.$
We define $F: M^b_+(\Omega) \to M^b_+(\Omega)$ by
(\ref{eq:final-size-eq-intro}),
$
F(w)(x) = \int_\Omega f(w(\xi)) \kappa(d \xi,x),
$
$x \in \Omega$, $ w \in M^b_+(\Omega)$.

\begin{proposition}[{\cite[Prop.4.1.5]{Dud}}]
\label{re:integr-nonneg-measur-integrable}
 If $g:\Omega \to \R_+$ is a measurable function,  then there exists
an increasing sequence of simple nonnegative functions such that
$g = \lim_{n\to \infty} g_n$ pointwise,
and $g$ is integrable.
If $g$ is bounded, the convergence is uniform.
\end{proposition}

Since $\kappa$ is a measure kernel, $F$ maps
$M^b_+(\Omega)$ into itself by Proposition
\ref{re:integr-nonneg-measur-integrable}.
The recursive equation (\ref{eq:final-recursive})
takes the form
\begin{equation}
\label{eq:final-recursive-kernel}
w_{n+1} = F ( w_n)     +w^\circ, \quad n \in \Z_+,
\qquad w_0 = w^\circ \in M^b_+(\Omega), \quad x \in \Omega.
\end{equation}
By induction, since $\kappa$ is a measure kernel and $f$ is increasing,
$(w_n)$ is an increasing sequence in $M^b_+(\Omega)$ which is bounded by $\kappa(\Omega,\cdot) + w^\circ$.
   The final cumulative force of infection, $w$
   is the pointwise limit of $(w_n)$,
\begin{equation}
w(x) = \lim_{n\to \infty} w_n(x), \qquad x \in \Omega.
\end{equation}
   By Beppo Levi's monotone convergence theorem, $w \in M^b_+(\Omega)$ satisfies
\begin{equation}
\label{eq:cum-FoI-final-ker}
w =w^\circ +   F(w).
\end{equation}

\begin{remark}
We introduce the measure kernel notation less to have shorter
formulas but to make the connection to
the final size consideration in \cite[Sec.3]{Thi77}.
\end{remark}

Not surprisingly, the final size of the
cumulative force of infection depends
on the final size of the cumulative initial
force of infection in an increasing way.

\begin{theorem}
\label{re:final-size-increase}
Let $w^\circ, \tilde w^\circ \in M^b_+(\Omega)$ and $w, \tilde w \in M^b_+(\Omega)$ be the minimal solutions
of
\[
w = w^\circ + F(w) \quad \hbox{ and } \quad
\tilde w = \tilde w^\circ + F(\tilde w).
\]
Then, if $w^\circ \le \tilde w^\circ$ on $\Omega$, also
$w \le \tilde w$ on $\Omega$ and $w - w^\circ \le
\tilde w - \tilde w^\circ$ as well.
\end{theorem}

\begin{proof}
Both $w = \lim_{n\to\infty} w_n $
and $\tilde w = \lim_{n\to \infty} \tilde w_n$
pointwise on $\Omega$,
where $w_n$ is given by the recursion
(\ref{eq:final-recursive-kernel})
and $\tilde w_n$ by an analogous recursion.
By induction, since $f$ is increasing, $w_n \le \tilde w_n$ on $\Omega$
for all $n \in \N$ and so $w \le \tilde w$
on $\Omega$. The last inequality follows
from $w - w^\circ = F(w) \le F(\tilde w) = \tilde w -\tilde w^\circ$.
\end{proof}

We continue this section with the following
observation which is as trivial as it is fundamental.

\begin{theorem}
$w = w_0$ if and only if $w_1=w_0$.
\end{theorem}

\begin{proof}
Let $w = w_0$. Since $w \ge w_1 \ge w_0$, this
implies $w_1 = w_0$.

Let $w_1 =w_0$. By induction and (\ref{eq:final-recursive-kernel}),
$w_n =w_0$ for all $n \in \N$ and so $w=w_0$
by (\ref{eq:cum-FoI-final-ker}).
\end{proof}

\begin{corollary}
\label{re:no-spread}
$w= w_0$ if and only if
\[
\int_\Omega   w_0(\xi)  \kappa(d \xi,x) =0,
\qquad x \in \Omega.
\]
\end{corollary}

Let $\cR_0 = \br(\kappa) >0$  and $\theta$
be the eigenfunctional of the operator $K$
associated with its spectral radius
(Theorem \ref{re:Krein}).

\begin{lemma}
\label{re:theta-inherit}
If $\theta(w^\circ)= 0$ and $(w_n)$ is provided by the recursion
the recursion (\ref{eq:final-recursive-kernel}),
then $\theta(w_n) =0$ for all $n \in \N$.
\end{lemma}

Unfortunately, since $\theta $ may not be continuous with respect to pointwise convergence,
this may not imply that $\theta(w) =0$ for
the pointwise limit $w$ of $(w_n)$.

If there is a positive minimum latency
period, we have the following without an extra
assumption.

\begin{theorem}
Let $a_0 >0$ and $\eta(x, a, \xi) =0 $ for all $a \in [0,a_0)$. Then, if $\theta(w_0) =0$, $\theta (J(t,\cdot))
=0 $ for all $t \ge 0$ and the minimal solution $J$
to (\ref{eq:cum-inf-force-eq}).
\end{theorem}

\begin{proof}
By Theorem \ref{re:pos-min-lat},
\[
J(t,x) \le w_n(x), \quad t \in [0, n a_0), \quad
x \in \Omega.
\]
Since $\theta$ is additive, by Lemma \ref{re:theta-inherit}
\[
\theta (J(t,\cdot)) \le \theta(w_n) =0, \quad t \in [0, n a_0), \quad
x \in \Omega.
\]
 Since this holds for all $n \in \N$, $\theta (J(t,\cdot))=0 $ for all $t \ge 0$.
\end{proof}

%%%%%%%%%%%%%%%%%%%%%%%%%%%%%%%%%%%%%%%%%%%%%

\subsection{Sequences of minimal solutions}

%%%%%%%%%%%%%%%%%%%%%%%%%%%%%%%%%%%%%%%%%%%%%%%%

Here are first results that relate
sequences of  minimal
solutions to $w = F(w) + w^\circ$ to fixed points of $F$.

\begin{theorem}
\label{re:final-size-converge}
Let  $(w^\circ_\ell)_{\ell \in \N}$ be a sequence
in $ M^b_+(\Omega)$
and $w^\circ_{\ell } \to 0$ as $\ell\to \infty$ pointwise
on $\Omega$.
For all $\ell \in \N$, let $w_\ell \in M^b_+(\Omega)$ be the minimal
solutions of
\[
w_\ell = F(w_\ell) + w^\circ_{\ell}.
\]
Then there exist a minimal  $\tilde w \in M^b_+(\Omega)$ such that $\tilde w = F(\tilde w)$ and
\[
\limsup_{n\to \infty} w_\ell(x) \le \tilde w(x),
\qquad x \in \Omega.
\]
\end{theorem}

\begin{proof}
Set
\[
\breve w (x) = \limsup_{n \to \infty} w_\ell(x),
\qquad x \in \Omega.
\]
By Fatou's Lemma, applied for any $x \in \Omega$,
\[
\breve w(x) \le
\int_\Omega \limsup_{n \to \infty} f(w_\ell(\xi))
\kappa(d\xi ,x) + \limsup_{n \to \infty } w^\circ_\ell(x).
\]
Since $f$ is monotone and continuous and $w^\circ_\ell \to 0$ as $\ell \to \infty$
pointwise on $\Omega$,
\[
\breve w(x) \le
\int_\Omega  f(\breve w_\ell(\xi))
\kappa(d\xi ,x)= F(\breve w)(x).
\]
We define recursively
\[
\tilde w_{n +1} = F(\tilde w_n), \quad  n \in \Z_+,
\qquad
\tilde w_0 = \breve w.
\]
Since $F$ is an increasing map, by induction,
$(\tilde w_n)$ is an increasing sequence of
functions that is also bounded and has a
pointwise limit $\tilde w \in M^b_+(\Omega)$,
$\tilde w \ge \breve w$.
By Lebesgue's theorem of dominated convergence,
$\tilde w = F(\tilde w)$.

We claim that $\tilde w$ is the minimal solution
of $\check w = F(\check w)$ with $\check w \ge \breve w$. Let $\check w$ be such a solution.
Then $\check w \ge \tilde w_0$. If $n \in \N$
and $\check w \ge \tilde w_n$, then
\[
\check w = F(\check w) \ge F(\tilde w_n)
= \tilde w_{n+1}.
\]
By induction, $\check w \ge \tilde w_n$ for all
$n \in Z_+$ and $\check w \ge \tilde w$ by taking
the pointwise limit.
\end{proof}

The next two results give us some vague idea of the
final size of the epidemic if the number of initial
infectives is small.

\begin{corollary}
\label{re:final-size-converge-zero}
Let  $(w^\circ_\ell)_{\ell \in \N}$ be a sequence
in $ M^b_+(\Omega)$
and $w^\circ_{\ell } \to 0$ as $\ell\to \infty$ pointwise
on $\Omega$.
For all $\ell \in \N$, let $w_\ell \in M^b_+(\Omega)$ be the minimal
solutions of $w_\ell = F(w_\ell) +w^\circ_{\ell}$.

Assume that  $\tilde w =0$ is the only
solution to $F(\tilde w) = \tilde w \in M^b_+(\Omega)$.

Then  $w_\ell \to 0$ as $\ell \to \infty$ pointwise on $\Omega$.
\end{corollary}

\begin{theorem}
\label{re:decreasing-converge}
Let $(w^\circ_\ell)$ be a
decreasing sequence in $M^b_+(\Omega)$,
 $w^\circ_\ell \to 0$ as $\ell \to \infty$ pointwise
on $\Omega$.
Let $(w_\ell)$ be the sequence of minimal solutions
of $w_\ell = F(w_\ell) + w^\circ_\ell$.

 Then there exists some fixed point
$\tilde w = F(\tilde w)$ in $M^b_+(\Omega)$
such that   $w_\ell  \searrow \tilde w$ and
$w_\ell -  w^\circ_\ell \searrow \tilde w$  as $\ell \to \infty$
pointwise on $\Omega$.
\end{theorem}

\begin{proof} By Theorem \ref{re:final-size-increase}, the sequences $(w_\ell)$ and $(w_\ell - w^\circ_\ell)$
 are decreasing. Since they are bounded below by the zero function, they converge pointwise to some $\tilde w \in M^b_+(\Omega)$ with
\[
 w_\ell - w^\circ_\ell \ge \tilde w, \qquad \ell \in \N.
\]
By Lebesgue's dominated convergence theorem,
applied to
\[
w_\ell(x) = \int_\Omega f(w_\ell(\xi)) \kappa(d \xi,x) + w^\circ_\ell(x)
\]
for each $x \in \Omega$, we take the limit as $\ell
 \to \infty$, and  $\tilde w$ satisfies $\tilde w = F(\tilde w)$.
\end{proof}

%%%%%%%%%%%%%%%%%%%%%%%%%%%%%%%%%%%%%%%%%%%%%%%%%

%%%%%%%%%%%%%%%%%%%%%%%%%%%%%%%%%%%%%%%%%%%%%%%%%%%%%%%

\section{Dominated
measure kernels}
\label{sec:intermed-thres}

A measure kernel $\kappa$ is called {\em dominated}
by $0\ne \nu \in \cM_+(\Omega)$
if
\begin{equation}
\label{eq:kernel-dom}
\kappa(\omega,x) \le \nu(\omega), \qquad \omega \in\cB.
\end{equation}
The measure kernel defined by (\ref{eq:meas-ker-susc}) is dominated by
a multiple of $S_0$
if the function $k$ is bounded on $\Omega\times \Omega$.

Throughout this section, we assume that $\cR_0 = \br(\kappa) = \br (K) >0$ and that $\theta$ is
the bounded linear eigenfunctional of $K$ associated with $\cR_0$
by Theorem \ref{re:Krein}.

\begin{proposition}
\label{re:dominated-unicon}
Let $\kappa$ be a dominated measure kernel.
Let $w$ be the minimal solution of $w =F(w) +w^\circ$ and $(w_n)$ be the recursion
$w_n = F(w_{n-1}) + w^\circ$, $n \in \N$, $w_0 =w^\circ$.

Then $\|w -  w_n\|_\infty \to 0$ and $\theta(w_n)
\to \theta(w)$ as $n \to \infty$.
\end{proposition}

\begin{proof}
Since $w_n \nearrow w $ pointwise on $\Omega$
and $f$ is continuous and increasing, $f\circ w_n \nearrow
f \circ w $ pointwise on $\Omega$. By Beppo Levi's theorem
of monotone convergence,
\[
\int_\Omega (f \circ w_n)
  d \nu \to \int_\Omega (f \circ w ) d \nu,
  \qquad n \to \infty .
\]
Since $f \circ w_n \le f \circ w$,
\[
 \int_\Omega |(f(w(\xi)) - f(w_n(\xi) )| \nu(d \xi) =
\int_\Omega \big((f(w(\xi)) - f(w_n(\xi) )\big )\nu(d \xi)
\stackrel{n \to \infty}{\longrightarrow} 0.
\]
By (\ref{eq:kernel-dom}),
\[
\begin{split}
\|w - w_{n +1}\|_\infty = &
\sup_{x\in \Omega}\int_\Omega \big |(f(w(\xi)) - f(w_n(\xi)) \big| \kappa(d\xi,x)
\\
\le &
 \int_\Omega \big ((f(w(\xi)) - f(w_n(\xi) )\big ) \nu(d \xi) \stackrel{n \to \infty}{\longrightarrow} 0.
 \qedhere
\end{split}
\]
\end{proof}

Lemma \ref{re:theta-inherit} implies the
following result.

\begin{proposition}
Let $w$ be the minimal solution of $w =F(w) +w^\circ$.
Let the measure kernel $\kappa$ be dominated.
Then $\theta(w) =0$ if $\theta(w^\circ) =0$.
\end{proposition}

The next result gives us some better idea than before
(Theorem \ref{re:decreasing-converge})
about the final size of the epidemic when the number
of initial infectives is small.

\begin{theorem}
\label{re:dominated-converge}
Let $\kappa$ be dominated
by a measure $\nu$ and let $(w^\circ_\ell)$ be a
decreasing sequence in $M^b_+(\Omega)$,
 $w^\circ_\ell \to 0$ as $\ell \to \infty$ pointwise
on $\Omega$.
Let $(w_\ell)$ be the sequence of minimal solutions
of $w_\ell = F(w_\ell) + w^\circ_\ell$.

\begin{itemize}
\item[(a)] Then $w_\ell - w^\circ_\ell  \searrow \tilde w$ as $\ell \to \infty$
uniformly on $\Omega$
 for some solution $\tilde w = F(\tilde w)$ in $M^b_+(\Omega)$.

\item[(b)] Assume in addition that   $\cR_0 > 1$ and $\theta(w^\circ_\ell) > 0$ for
all $\ell \in \N$.

Then %$w_\ell - w^\circ_\ell \searrow \tilde w$
%uniformly on $\Omega$ and
  $\int_\Omega (f \circ \tilde w) d\nu \ge \ln \cR_0$
and $\|\tilde w\|_\infty \ge \ln \cR_0$.
\end{itemize}
\end{theorem}

\begin{proof}
(a)
By Theorem \ref{re:final-size-increase}, the sequences $(w_\ell)$ and $(w_\ell - w^\circ_\ell)$
 are decreasing.
By Theorem  \ref{re:decreasing-converge},
$ w_\ell - w^\circ_\ell \to \tilde w$
as $\ell \to \infty$ pointwise on $\Omega$
for some $w \in M^b_+(\Omega)$ with $\tilde w
= F(\tilde w)$.
For all $\ell \in \N$ and $x \in \Omega$,
\[
0 \le w_\ell (x)  - w^\circ_\ell(x) - \tilde w(x)
=
\int_\Omega \big ( f(w_\ell(\xi))- f(\tilde w(\xi) \big )\kappa(d\xi,x).
\]
Since $\kappa$ is dominated by $\nu$,
\[
0 \le w_\ell (x)  - w^\circ_\ell(x) - \tilde w(x)
\le
\int_\Omega \big ( f(w_\ell(\xi))- f(\tilde w(\xi) \big ) \nu(d\xi).
\]
Since the right hand side does not depend on $x$,
$\| w_\ell - w^\circ_\ell - \tilde w \|_\infty
\to 0$ as $\ell \to \infty$.

(b) Let $\cR_0 > 1$.
For all $\ell \in \N$, since $\kappa$ is dominated by $\nu$,
\[
0 \le w_\ell(x) - w^\circ_\ell(x)
\le
\int_\Omega (f \circ w_\ell) d \nu
\]
and, by Theorem \ref{re:epi-occur},
\[
\ln \cR_0 \le \|w_\ell - w^\circ_\ell\|_\infty
\le
\int_\Omega (f \circ w_\ell) d \nu.
\]
By Lebesgue's theorem of dominated convergence,
\[
\int_\Omega (f \circ w ) d \nu \ge
 \ln \cR_0.
 \]
Further $\|\tilde w \|_\infty \ge \ln \cR_0$.
\end{proof}

\begin{proof}[Proof of Theorem \ref{re:dominated}]
Apply Theorem \ref{re:dominated-converge}
with  $w^\circ_\ell = (1/\ell)w^\circ$
for $\ell \in \N$ and $w = w_1$.
\end{proof}

%%%%%%%%%%%%%%%%%%%%%%%%%%%%%%%%%%%%%%%%%%%%%

\subsection{Positivity points of $\theta$}

In view of the previous results, it is of interest
for which $w^\circ \in M^b_+(\Omega)$ we have
$\theta(w^\circ) > 0$.

\begin{proposition}
\label{re:theta-pos}
Let $\cR_0 >0$ and let there exist $\nu_j, \tilde \nu_j \in \cM_+(\Omega)$ and $k_j \in M^b_+(\Omega)$,
$j=1,\ldots,n$, such that
\begin{equation}
\label{eq:sum-inequ}
\sum_{j=1}^n \tilde \nu_j(\omega) k_j(x) \le \kappa(\omega, x) \le \sum_{j=1}^n \nu_j(\omega) k_j(x),
\qquad \omega \in \cB, \quad x \in \Omega.
\end{equation}
Then, for any $w \in M^b_+(\Omega)$, $\theta (w) > 0$ if $\int_\Omega w \, d \tilde \nu_j > 0$ for $j=1,\ldots, n$, while
$\theta(w) =0$ if  $\int_\Omega w \, d  \nu_j = 0$ for $j=1,\ldots, n$.
\end{proposition}

\begin{proof} By (\ref{eq:sum-inequ}), for $w \in M^b_+(\Omega)$,
\[
\sum_{j=1}^n \Big (\int_\Omega w \, d \tilde \nu_j \Big ) k_j
\le K w \le
\sum_{j=1}^n \Big (\int_\Omega w \, d \nu_j \Big ) k_j.
\]
Then $Kw =0$ and $\theta(w) =0$ if
$\int_\Omega w d \nu_j=0$ for $j =1, \ldots, j$.

Recall the constant function $u_1$ with value 1. Since $\theta$ is linear,
\[
0 < \cR_0 \theta (u_1) = \theta(K u_1)
\le
\sum_{j=1}^\infty \nu_j(\Omega) \theta(k_j).
\]
This implies that $\sum_{j=1}^n \theta(k_j) > 0$.
Further
\[
\cR_0 \theta(w) \ge \sum_{j=1}^n \Big (\int_\Omega w d \tilde \nu_j \Big ) \theta(k_j).
\]
So $\theta (w) >0$ if $\int_\Omega w \, d \tilde \nu_j >0$ for $j =1, \ldots, n$.
\end{proof}

%%%%%%%%%%%%%%%%%%%%%%%%%%%%%%%%%%%%%%%%%%%%%%%%%%

\section{Semi-separable measure kernels}
\label{sec:semi-sep}
%%%%%%%%%%%%%%%%%%%%%%%%%%%%%%%%%%%%%%%%%%%%%%%%%%%%

A measure kernel $\kappa$ is called {\em semi-separable} if there are nonzero $k_0 \in M^b_+(\Omega)$,
$\nu \in  \cM_+(\Omega)$ and $\delta \in (0,1]$ such
that
\begin{equation}
\label{eq:separ}
\delta  \nu (\omega) k_0(x)\le \kappa(\omega,x)
\le
\nu(\omega) k_0(x), \qquad x \in \Omega, \quad
\omega \in \cB.
\end{equation}

\begin{proposition}
\label{re:theta-pos-separ}
Assume that $\kappa$ is semi-separable, (\ref{eq:separ}), and let $w \in M^b_+(\Omega)$.
Then, $\cR_0 > 0$ if and only if $\int_\Omega k_0 \, d \nu >0$, and $\theta (w) >0$ if and only if
$\int_\Omega w \ d\nu > 0$.
\end{proposition}

\begin{proof}
By (\ref{eq:separ}),
\[
\delta \Big (\int_\Omega k_0 \,d \nu\Big) k_0 \le K k_0 \le \Big (\int_\Omega k_0 \,d \nu\Big) k_0.
\]
This implies \cite[Thm.3.1]{Thi17}\cite[Thm.5.31]{Thi24} \cite[Thm.3.3]{Thi17}\cite[Thm.6.15]{Thi24} that
\[
\delta \int_\Omega k_0 \,d \nu \le \cR_0 =
\br(K) \le \int_\Omega k_0 \,d \nu.
\]
The remaining statement is a special case of
Proposition \ref{re:theta-pos} for $n =1$.
\end{proof}

\begin{proposition}
\label{re:unforced-unique}
Assume that $\kappa$ is semi-separable, (\ref{eq:separ}).

\begin{itemize}
\item[(a)] Then there exists at most one non-zero solution
$w \in M^b_+(\Omega)$ of $w $ $= F(w)$.

\item[(b)] If $\cR_0 \le 1$, the zero function is the
only $w \in M^b_+(\Omega)$ with $w = F(w)$.
\end{itemize}
\end{proposition}

\begin{proof}
We follow \cite[Sec.6.1]{Kra}.
By (\ref{eq:separ}),
\begin{equation}
\label{eq:separ-compare}
\delta k_0(x) \int_\Omega (f \circ w) d \nu \le F(w)(x) \le k_0(x) \int_\Omega (f \circ w) d \nu
\le k_0(x) \nu(\Omega).
\end{equation}
By (\ref{eq:separ-compare}), $F(w)$ is not the zero function if and only if
\begin{equation}
\label{eq:separ-compare-int}
\int_\Omega (f \circ w) d \nu >0.
\end{equation}
Since $f$ is concave and $f(0) =0$,
\[
F(tw) \ge t F(w), \qquad t \in [0,1], \quad w \in M^b_+(\Omega).
\]
(a) According to \cite[Thm.6.3]{Kra}, it is sufficient
to show that, for any $t \in (0,1)$ and any
$w \in M^b_+(\Omega)$ with
\[
\int_\Omega (f \circ w) d \nu >0,
\]
some $\eta >0$ can be found such that
\begin{equation}
\label{eq:strict-sublin}
F(t w) \ge (1+\eta) tF(w).
\end{equation}
Since $f$ is strictly concave,
$f(tr) > t f(r)$ for all $r > 0$, $t \in (0,1)$.
Let $t \in (0,1)$. Then $f(t w(\xi)) - t f(w(\xi))$ is nonnegative for
all $\xi \in \Omega$ and is positive for $\xi
\in \Omega$ if $w(\xi) > 0$.

Suppose that
\[
\int_\Omega (f(t w(\xi)) - t f(w(\xi))) \nu(d \xi) =0.
\]
Then $w(\xi)=0$ for $\nu$-a.a. $\xi \in \Omega$
and $\int_\Omega (f \circ w) d \nu =0$.

By contraposition, if $\int_\Omega (f \circ w) d \nu > 0$,
\[
\tilde \eta = \int_\Omega \big (f(t w(\xi)) - t f(w(\xi))\big ) \nu(d \xi) > 0.
\]
Now,
\[
F(tw)(x) - t F(w)(x)
=
\int_\Omega \big [f(tw(\xi)) - t f(w(\xi))\big ] \kappa (d\xi,x)
\]
with the expression in $[\cdot]$ being nonnegative.
By (\ref{eq:separ}),
\[
F(tw)(x) - t F(w)(x)
\ge
\delta k_0(x) \int_\Omega \big [f(tw(\xi)) - t f(w(\xi))\big ] \nu(d \xi)
\ge
\delta k_0(x) t \hat \eta
\]
with $\hat \eta = \tilde \eta/t$.
By (\ref{eq:separ-compare}) and (\ref{eq:separ}),
\[
F(tw)(x) - t F(w)(x)
\ge \delta k_0(x) t \hat \eta \frac{1}{\nu(\Omega)}\int_\Omega f(w(\xi)) \nu(d\xi)
\ge
\delta t  \frac{\hat \eta}{\nu(\Omega)} F(w(x)).
\]
We reorganize,
\[
F(tw) (x) \ge \Big (1+ \frac{\delta \hat \eta}{\nu(\Omega)}\Big )
t Fw(x), \qquad x \in \Omega,
\]
and we have shown (\ref{eq:strict-sublin})
with $\eta =\frac{\delta \hat \eta}{\nu(\Omega)}$.

(b) Suppose that $0 \ne w \in M^b_+(\Omega)$ satisfies $w = F(w)$.
Then
(\ref{eq:separ-compare-int}) holds. Let $t =1/2$.
By the same considerations as before, there
exists some $\eta >0$ such that (\ref{eq:strict-sublin}) is valid. By (\ref{eq:order-der}),
\[
K (tw) \ge F(tw) \ge (1+\eta) (tw).
\]
This implies that $\cR_0 = \br(K) \ge (1+\eta)$ \cite[Thm.3.1]{Thi17} \cite[Thm.5.31]{Thi24},
a contradiction.
\end{proof}

We are now in the position to prove
the preview results in  Section
\ref{subsec:semisep-intro}.

\begin{proof}[Proof of Theorem \ref{re:semisep-intro}]
Part (a) follows from Corollary \ref{re:no-spread}
and (\ref{eq:separ}).

(b) Existence of $\tilde w$ follows from
Corollary \ref{re:dominated-fixp}
with $\nu$
being replaced by $\|k_0\|_\infty \nu$.
Uniqueness of $\tilde w$ follows from Proposition \ref{re:unforced-unique}.

The remaining statements follow from
Proposition \ref{re:theta-pos-separ},
Theorem \ref{re:epi-occur} (b) and
Theorem \ref{re:dominated}.
\end{proof}

\begin{proof}[Proof of Theorem \ref{re:semisep-sequ-intro}]
(a) Pointwise convergence of $w_\ell \to 0$ follows from Corollary \ref{re:final-size-converge-zero} and
Proposition \ref{re:unforced-unique} (b).
For all $\ell \in \N$,
\[
0 \le w_\ell(x) -w^\circ_\ell (x)
\le
k_0(x) \int_\Omega f(w_\ell(\xi) ) \nu(d\xi)
\to 0
\]
as $\ell \to \infty$ by Lebesgue's theorem
of dominated convergence
because $f(w_\ell(\xi) )$
$\to 0$ as $\ell \to \infty$.
Since $k_0$ is bounded, $w_\ell - w^\circ_\ell
\to 0$ as $\ell \to \infty$ uniformly on $\Omega$.

\smallskip

(b) By Fatou's lemma and the increase
and continuity of $f$,
\[
\int_\Omega f \circ (\limsup_{\ell \to \infty} w_\ell) d \nu
\ge
 \int_\omega \limsup_{\ell \to \infty} (f \circ w_\ell) d\nu
 \ge
 \limsup_{\ell \to \infty} \int_\Omega (f \circ w_\ell ) d \nu.
 \]
 By Theorem \ref{re:semisep-intro} (b),
 \[
 \int_\Omega (f \circ w_\ell ) d \nu\ge \ln \cR_0/\|k_0\|_\infty, \qquad \ell \in \N.
 \]
So, $\limsup_{\ell \to \infty} w_\ell$ is not
the zero function. By Theorem \ref{re:final-size-converge},
there exists a solution
$\hat w$ of $\hat w = F(\hat w)$ with
 $\limsup_{\ell \to \infty} w_\ell \le \hat w$.
 Then $\hat w$ is not the zero function and
 $\hat w = \tilde w$ by Theorem \ref{re:semisep-intro} (b). Since $w_\ell \ge \tilde w$
 for all $\ell \in \N$ by Theorem \ref{re:semisep-intro} (b),
 $w_\ell \to \tilde w = \hat w$ as $\ell \to \infty$ pointwise on $\Omega$.

 For all $\ell \in \N$,
 \[
 0 \le w_\ell(x) - w^\circ_\ell(x) - \tilde w(x)
 =
 \int_{\Omega} \big (f(w_\ell(\xi)) - f(\tilde w(\xi)) \big ) \kappa(d \xi, x).
 \]
 By (\ref{eq:separ}),
 \[
 0 \le w_\ell(x) - w^\circ_\ell(x) - \tilde w(x)
\le
k_0(x) \int_\Omega \big (f(w_\ell(\xi)) - f(\tilde w(\xi)) \big ) \nu(d \xi).
\]
Since $k_0$ is bounded and $f(w_\ell(\xi)) - f(\tilde w(\xi)) \to 0$ as $\ell \to \infty$
pointwise for $\xi \in \Omega$,
by Lebesgue's theorem of dominated convergence,
\[
w_\ell(x) - w^\circ_\ell(x) \to \tilde w(x),
\quad
\ell \to \infty,
\]
uniformly for $x \in \Omega$.
 \end{proof}

%%%%%%%%%%%%%%%%%%%%%%%%%%%%%%%%%%%%%%%%%%%%%%%%%

\section{Metric spaces of traits and Feller kernels}
\label{sec:metric}

%%%%%%%%%%%%%%%%%%%%%%%%%%%%%%%%%%%%%%%%%%%%%

To replace semi-separability of the kernel
as an assumption, we assume that $\Omega$ is a metric space with metric $\rho$ and $\cB$ the $\sigma$-algebra of
Borel sets. Assume that $\Omega$ is not just
a single point. Let $C^b(\Omega)$ denote the Banach
space of bounded continuous functions with the
supremum norm which is a closed subspace of $M^b(\Omega)$.
We start with a few technical observations.

\begin{remark}
\label{re:regular-measure}
If $\mu$ is a finite nonnegative
measure on $\cB$, then $\mu$ is inner and outer
regular and $C^b(\Omega)$ is dense in $L^1(\Omega, \mu)$.
\end{remark}

\begin{proof}
See \cite[Sec.12.1]{AlBo}  for the regularity statement
and \cite[L.IV.8.19]{DS} for the density statement. Notice that $\mu$ is regular in \cite[L.IV.8.19]{DS} if and and only it is inner and outer regular in
\cite[Sec.12.1]{AlBo}.
\end{proof}

\begin{lemma}
\label{re:top-pos}
Let $g \in M^b_+(\Omega)$ be topologically positive
(Definition \ref{def:Feller}).
\begin{itemize}
\item[(a)] Then there exists a non-zero Lipschitz continuous function $\tilde g:\Omega \to \R_+$
with a Lipschitz constant $\le 1$ such that $\tilde g \le g$ and $\tilde g$ is strictly positive on
every nonempty open subset $U$ with $\inf_U g >0$.

\item[(b)] If $g$ is lower semicontinuous, $x \in \Omega$ and $g(x) >0$, then $\tilde g(x) >0$ for the function $\tilde g$ from (a).
\end{itemize}
\end{lemma}

\begin{proof} (a) Recall that $\rho$ denotes the metric
on $\Omega$. Let $g \in M^b_+(\Omega)$.
Define
\[
\tilde g(x) = \inf \big \{ \rho(x,\xi) + g(\xi);\; \xi \in \Omega \big \}, \qquad x \in \Omega.
\]
Then $\tilde g$ is Lipschitz continuous with a Lipschitz constant $\le 1$and $0 \le \tilde g \le g$ on $\Omega$ \cite[Prop.2.78]{Thi24}.

Let $U$ be a nonempty open subset of $\Omega$ such that $\inf_U g >0$. Then,  for any $x \in U$ there exists some $\delta \in (0, \inf_U g)$
such that $\xi \in U$ and $g(\xi) > \delta$ whenever $\xi \in \Omega$ and $\rho(x,\xi) < \delta$.
This implies that
\[
\rho(x,\xi) + g(\xi)\ge \delta, \quad \xi \in \Omega,
\]
and $\tilde  g(x) \ge \delta$.

(b) Let $g$ be lower semicontinuous.
Suppose that $x \in \Omega$ and $\tilde g(x) =0$.
Then there exists a sequence $(\xi_n)$ in $\Omega$
such that $\rho(x, \xi_n) + g(\xi_n) \to 0$
as $n \to \infty$. In particular, $\xi_n \to x $
as $n \to \infty$. Since $g$ is lower continuous,
$0=\liminf_{n\to \infty} g(\xi_n) \ge g(x)$
\cite[L.A.49]{Thi24}.

Conversely, $g(x) >0$
implies that $\tilde g(x) >0$.
\end{proof}

\subsection{Feller kernels}
\label{subsec:Feller}

A measure kernel $\kappa$ is called a {\em Feller
kernel} if the map $K$ on $M^b(\Omega)$ induced by
$\kappa$ maps $C^b(\Omega)$ into itself
(Definition \ref{def:Feller}).

\begin{proposition}
\label{re:Feller-semicon}
Let $\kappa$ be a Feller kernel and $g: \Omega
\to \R_+$ be lower semicontinuous. Then $Kg$
and $Fg$ are lower semicontinuous.
\end{proposition}

\begin{proof}
Let $g: \Omega
\to \R_+$ be lower semicontinuous.
Then $g$ is the pointwise limit of an
increasing sequence of Lipschitz continuous
functions \cite[Thm.3.13]{AlBo}\cite[Prop.2.78]{Thi24} and
$f \circ g$ is the pointwise limit of an
increasing sequence of continuous functions.
Since $\kappa$ is a Feller kernel, $Kg$
and $F(g)$ are pointwise limits of increasing
sequences of continuous functions and thus
lower semicontinuous \cite[L.A.53]{Thi24}.
\end{proof}

\begin{theorem}
\label{re:min-sol-lowsemi}
 Let $\kappa$ be a Feller
kernel and $w^\circ \in M^b_+(\Omega)$
be lower semicontinuous. Then the minimal solution
$w \in M^b_+(\Omega)$ of $w = F(w) + w^\circ$
is lower semicontinuous.
\end{theorem}

\begin{proof}
Recall that $w$ is given as the pointwise limit $w = \lim_{n \to
 \infty} w_n$ of the recursion
\begin{equation}
\label{eq:recursion-intro1}
w_{n} = F(w_{n-1}) + w^\circ, \quad  n \in \N,
\qquad
w_0 = w^\circ.
\end{equation}
By induction, $(w_n)$ is a increasing sequence
of functions in $M^b_+(\Omega)$ and every $w_n$
is lower semicontinuous by Proposition \ref{re:Feller-semicon} and so is the pointwise
limit $w= \sup_{n \in \N} w_n$ \cite[L.A.53]{Thi24}. \end{proof}

\begin{proposition}
Let $\kappa$ be a measure kernel
such that $\kappa(\Omega,\cdot)$ is continuous
on $\Omega$ and $Kg$ is lower semicontinuous
for any $g \in C^b_+(\Omega)$.
Then $\kappa$ is a Feller kernel.
\end{proposition}

\begin{proof}
Let $g \in C^b_+(\Omega)$. It is sufficient
to show that $K g$ is upper semicontinuous.
Define $\tilde g \in C^b_+(\Omega)$ by
$\tilde g(x)  = \|g\|_\infty - g(x)$, $x \in \Omega$. By assumption, $K \tilde g$ is lower semicontinuous,
\[
K \tilde g = \|g\|_\infty \kappa(\Omega, \cdot)
- K g.
\]
So, $Kg$ is upper semicontinuous. See
\cite[Rem.A.48]{Thi24}.
\end{proof}

%%%%%%%%%%%%%%%%%%%%%%%%%%%%%%%%%%%%%%%%%%%%%%%%%%

\subsection{Topological irreducibility}

We call $\kappa$ {\em topologically irreducible}
if for any nonempty open strict subset $U$ of
$\Omega$ there exist some $x \in \Omega \setminus U$ such that $\kappa (U,x) > 0$. If $\kappa$
is topologically irreducible, then \cite[Rem.13.57]{Thi24}
\begin{equation}
\label{eq:top-irr-Omega}
\kappa\big ( \Omega \setminus \{x\}, x\big) > 0,
\qquad
x \in \Omega.
\end{equation}

For the special  case (\ref{eq:meas-ker-susc}), $\kappa$ is topologically
irreducible if $S_0(\omega) >0$ for any nonempty open subset of $\Omega$ and if for any nonempty
open strict subset $\omega$ of $\Omega$ there
exists some $x \in \Omega \setminus \omega$
such that $k(x, \cdot)$ is not zero a.e. on
$\omega$.

\begin{theorem}
\label{re:irred-pos}
Let $\kappa$ be topologically irreducible,
$w^\circ$ be lower semicontinuous and not the zero function and
$w$ be the minimal solution of $w = F(w) + w^\circ$ (\ref{eq:final-limit}).
Then $w$ is lower semicontinuous and  $w-w^\circ$ is strictly positive on $\Omega$.
\end{theorem}

\begin{proof}
By Proposition \ref{re:min-sol-lowsemi},
$w$ is lower semicontinuous
and $U=\{w>0\}$ is open \cite[A.51]{Thi24}.

Since $w^\circ$ is not the zero function, $\{w>0\}$ is not the empty set.

Suppose that $U=\{w> 0\} \ne \Omega$. Since $\kappa$
is topologically irreducible,
there exist some $x \in \Omega \setminus U$ such that $\kappa (U,x) > 0$. Recall that
\[
w(x) = \int_\Omega f(w(\xi)) \kappa(d\xi,x)
+w^\circ(x).
\]
Now $U = \bigcup_{n \in \N} \{w > 1/n\}$.
Since $\kappa (\cdot,x)$ is a measure,
there is some $n \in \N$ such that
$\kappa(\{w > 1/n\}, x ) >0$.
Then
\[
w(x) \ge  f(1/n) \kappa(\{w> 1/n\}, x)> 0,
\]
and $x \in U$, a contradiction. Now $f\circ w$
is strictly positive and $F(w)=w -w^\circ$ is strictly positive by (\ref{eq:top-irr-Omega}).
\end{proof}

\begin{remark}
\label{re:irr-pos-nec}
Topological irreducibility of the kernel
is necessary for the epidemic to always  reach
all traits.
\end{remark}

\begin{proof}
Assume that $\kappa $ is not topologically irreducibility. Then there
exists some nonempty open strict subset $U$
of $\Omega$ such that $\kappa(U,x) =0$
for all $x \in \Omega \setminus U$.
Since $\Omega$ is a metric space, there is some
nonzero continuous function $w_0$ such that
$w_0 (x) =0$ for all $x \in \Omega \setminus U$
\cite[L.3.20]{AlBo}.
By induction and  (\ref{eq:final-recursive-kernel})
$w_n(x) =0$ for all $n \in \N$
and $x \in \Omega \setminus U$.
Then the pointwise limit function $w$ also satisfies $w(x) =0$
for all $x \in \Omega \setminus U$.
\end{proof}

\begin{proof}[Proof of Theorem \ref{re:irred-pos1}]
Let $U$ be an open nonempty subset of $\Omega$
such that $\inf_U w > 0$. By Lemma \ref{re:top-pos},
there exists some $\tilde w^\circ \in C^b_+(\Omega)$
which is strictly positive on $U$ and  $ 0 \le \tilde w^\circ \le w^\circ$
on $\Omega$.
Let $\tilde w$ be the minimal solution
of $\tilde w = F(\tilde w ) + \tilde w^\circ$. By Theorem \ref{re:final-size-increase},
 $\tilde w \le w$ on $\Omega$.
By Theorem \ref{re:irred-pos}, $\{\tilde w > 0\} = \Omega $
and so $\{w > 0\} = \Omega$. Then $f \circ w $
is strictly positive on $\Omega$
and $F(w)=w -w^\circ$ is strictly positive by (\ref{eq:top-irr-Omega}).
\end{proof}

\begin{theorem}
\label{re:irred-pos2}
Let $\tilde w: \Omega \to \R_+$ be a lower semicontinuous non-zero solution to the inequality
\[
\tilde w(x) \ge \int_\Omega f(\tilde w(\xi)) \kappa(d
\xi,x), \qquad x \in \Omega,
\]
and let $\kappa$ be a topologically irreducible
Feller kernel. Then $\tilde w$ is strictly positive
on $\Omega$.
\end{theorem}

\begin{proof}
Since $\tilde w$ is lower semicontinuous  and not
the zero function,
the set $U =\{\tilde w > 0\}$ is an open nonempty
set \cite[A.51]{Thi24}. If the assertion is false, $U$ is a strict
subset of $\Omega$. Since $\kappa$ is topologically
irreducible, there exists some $x \in \Omega \setminus U$ such that $\kappa(U,x) > 0$.
Since $\kappa (\cdot,x)$ is a measure,
there is some $n \in \N$ such that
$\kappa(\{w > 1/n\}, x ) >0$.
Then
\[
w(x) \ge  f(1/n) \kappa(\{w> 1/n\}, x)> 0,
\]
and $x \in U$, a contradiction.
\end{proof}

%%%%%%%%%%%%%%%%%%%%%%%%%%%%%%%%%%%%%%%%%%%%%%%%%%%%

\subsection{Tightness}

%%%%%%%%%%%%%%%%%%%%%%%%%%%%%%%%%%%%%%%%%%%%%%%

Recall Definition \ref{def:Feller} and \ref{def:tight}.

For instance, the Feller kernel $\kappa$
in (\ref{eq:meas-ker-susc}) is tight if $S_0$
is tight and $k$ is bounded on $\Omega^2$.
Recall that every finite nonnegative measure
on $\cB$ is tight  if $\Omega$ is a Polish
space, i.e., $\Omega$ is separable and
complete under a metric that is topologically
equivalent to the original one.

\begin{proof}[Proof of Theorem \ref{re:eigenmeasure-intro}]
 Existence of the tight eigenmeasure $\mu$
follows from \cite[Thm.13.39]{Thi24}
and \cite[Thm.13.42]{Thi24}.

In addition, let $\kappa$ be topologically irreducible. Apply \cite[Cor.13.60]{Thi24} and Lemma \ref{re:top-pos}.
\end{proof}

%%%%%%%%%%%%%%%%%%%%%%%%%%%%%%%%%%%%%%%%%%%%%%%%%%%%%

%%%%%%%%%%%%%%%%%%%%%%%%%%%%%%%%%%%%%%%%%%%%%%%%%%%%%%

\subsection{Strong Feller kernels}
\label{subsec:Feller-strong}

While Theorem \ref{re:tight-irred-pandemic}
in conjunction with Theorem \ref{re:epi-occur} (a)
gives a good display of the threshold properties
of $\cR_0$,  the relation
of minimal solutions of $w = F(w) + w^\circ$
to fixed points $\tilde w = F(\tilde w)$
may be informative. To this end, we strengthen
the concept of a Feller kernel (Definition \ref{def:strong-Fel}).

\begin{proposition}
\label{re:kappa-lower-sc}
Let $\kappa$ be a measure kernel.
Assume that $\kappa(\Omega,\cdot)$
is continuous on $\Omega$ and that
$\kappa(\omega,\cdot)$ is lower semicontinuous for any
$\omega \in \cB$.

Then $\kappa$ is a strong Feller kernel.
\end{proposition}

\begin{proof}
Let $\omega \in \cB$. Then $\kappa(\omega, \cdot)$
is lower semicontinuous. Further, $\Omega \setminus \omega \in \cB$.
 Then $\kappa(\Omega \setminus \omega, \cdot )$ is lower semicontinuous,
\[
\kappa(\Omega \setminus \omega, \cdot)=
\kappa (\Omega, \cdot) - \kappa(\omega, \cdot).
\]
Since $\kappa(\Omega, \cdot)$ is continuous, $\kappa (\omega, \cdot)$ is upper semicontinuous \cite[Rem.A.48]{Thi24}. Since $\kappa(\omega, \cdot)$ is both lower and upper semicontinuous,
it is continuous. Use Lemma A.49 and A.50
in \cite{Thi24}.
\end{proof}

 In the framework of the model in
Section \ref{sec:model} it is difficult to
find an example of a Feller kernel that is not
a strong Feller kernel. See Proposition
\ref{re:dominated-Feller-strong}.

In the context of \cite{Thi79}, it is easy
to give an example of Feller kernel that is not
a strong Feller kernel like
$\kappa(\omega,x) = \chi_\omega (x)$ where $\chi_\omega$ is the characteristic or indicator function
of $\omega$, $\chi_\omega (x) =1$ if $x \in \omega$
and 0 otherwise.

\begin{lemma}
\label{re:Feller-strong}
If $\kappa$ is a strong Feller kernel, then
$K$ and $F$ map $M^b_+(\Omega)$ into $C^b_+(\Omega)$.
\end{lemma}

\begin{proof}
Let $\kappa$ be a strong Feller kernel.
Since bounded measurable functions
are uniform limits of linear combinations
of characteristic functions (Proposition
\ref{re:integr-nonneg-measur-integrable}),
the operators $K$ and $F$ induced by a strong
Feller kernel map $M^b_+(\Omega)$ into $C^b_+(\Omega)$.
\end{proof}

In Section \ref{sec:intermed-thres}, we considered dominated measure kernels.
For perspective, we mention the following result.

\begin{proposition}
\label{re:dominated-Feller-strong}
Let $\kappa$ be a Feller kernel that is dominated
by some $\nu \in \cM_+(\Omega)$.
Then $\kappa$ is a strong Feller kernel.
Actually,
\[
K g \in C^b(\Omega), \qquad g \in L^1(\Omega, \nu).
\]
\end{proposition}

\begin{proof}
Since the Feller kernel $\kappa$ is dominated
by the measure $\nu$, for each $x \in \Omega$,
the measure $\kappa(\cdot,x)$ is absolutly
continuous with respect to $\nu$.
By the Radon-Nikodym theorem, there exists
some $k_x \in L^1_+(\Omega,\nu)$ such that
\[
\nu(\omega) \ge \kappa(\omega, x) = \int_\omega k_x(\xi) \nu (d\xi), \qquad \omega \in \cB, \quad x \in \Omega.
\]
Hence, $k_x(\xi) \le 1$ for all $\xi \in \Omega$.
By Remark \ref{re:regular-measure},
$C^b(\Omega)$ is dense in $L^1(\Omega, \nu)$.
Let $g \in L^1(\Omega, \nu)$. Then there exists
a sequence $(g_n)$ in $C^b(\Omega)$ such that
$\|g - g_n\|_1 \to 0$ as $n \to \infty$. For all
$x \in \Omega$,
\[
\big |K g_n(x) - Kg(x)\big|
\le
\int_\Omega |g_n(\xi) - g(\xi)| k_x(\xi) \, \nu(d \xi)
\le
\|g_n- g\|_1,
\]
and $Kg_n \to K g$ as $n \to \infty$ uniformly
on $\Omega$. Since $\kappa$ is a Feller kernel,
all $K g_n$ are continuous and so is their uniform limit $Kg$.
\end{proof}

\begin{theorem}
\label{re:tight-Feller-converge}
Let $\kappa$ be a tight strong Feller kernel.
Let $(w^\circ_\ell)$ be a decreasing sequence
 in $M^b_+(\Omega)$
that converges to 0 uniformly on all compact subsets  of $\Omega$.
For $\ell \in \N$, let
 $w_\ell$ be the minimal solution to
$w_\ell = F(w_\ell) + w^\circ_\ell$ in $M^b_+(\Omega)$.

Then there is some
$\tilde w \in C^b_+(\Omega)$ such that $\tilde w
= F(\tilde w)$
   and $w_\ell -w^\circ_\ell \to \tilde w$ uniformly on  $\Omega$,
\[
w_\ell \ge \tilde w + w^\circ_\ell, \qquad \ell \in \N.
\]
Further, the following holds:

 If $\cR_0 >1$ and $\kappa$ is topologically
irreducible and all functions $w^\circ_\ell$
are topologically positive, then $\tilde w $ is strictly positive
and
$\|\tilde w \|_\infty \ge \ln \cR_0$.

\end{theorem}

\begin{proof}
By Theorem \ref{re:final-size-increase},
$(w_\ell)$ and $(w_\ell - w^\circ_\ell)$
are decreasing sequences that converge
to some $\tilde w \in M^b_+(\Omega)$ pointwise
on $\Omega$. Since  $\kappa$ is a strong Feller kernel,
all $w_\ell - w^\circ_\ell = F(w_\ell)$ are continuous.
We have
$w_\ell \ge \tilde w + w^\circ_\ell$ for all
$\ell \in \N$. %\cite[L.A.52]{Thi24}.
By Lebesgue's theorem of dominated
convergence, $\tilde w = F(\tilde w)$ and
$\tilde w$ is continuous because $\kappa$
is a strong Feller kernel.

  By Dini's lemma, $(w_\ell - w^\circ_\ell)$
converges to $\tilde w$ uniformly on every compact
subset $W$ of $\Omega$. Since $w^\circ_\ell
\to 0$ as $\ell \to \infty$ uniformly
on every compact subset $W$ of $\Omega$,
$w_\ell \to \tilde w$ as $\ell \to \infty$
uniformly on every compact subset $W$
of $\Omega$. Since $f$ is uniformly
continuous,
\begin{equation}
\label{eq:tight-str-Feller-converge}
\sup_w (f \circ w_\ell - f \circ \tilde w)
\to 0, \qquad \ell \to \infty
\end{equation}
for every compact subset $W$ of $\Omega$.

Let $\epsilon >0$.
Since $\kappa$ is tight, there exists some
compact subset $W$ of $\Omega$ such that
$\kappa(\Omega \setminus W, x) \le \epsilon$
for all $x \in \Omega$. For all $\ell \in \N$
and $x \in \Omega$,
\[
\begin{split}
 0 \le & w_\ell(x) - w^\circ_\ell(x) -\tilde w(x)
 \\
  \le & \int_W \big ( f(w_\ell(\xi))
  - f(\tilde w(\xi)) \big ) \kappa(d \xi,
x)
+
\int_{\Omega \setminus W} f(w_\ell(\xi)) \kappa(d \xi,
x) .
\end{split}
\]
By the properties of $f$,
\[
\|w_\ell -  w^\circ_\ell - \tilde w\|_\infty
\le
\sup_W (f\circ w_\ell - f \circ \tilde w) \; \sup_{x \in \Omega}\kappa(\Omega,x)
+
\sup_{x\in \Omega} \kappa(\Omega\setminus W, x).
\]
By (\ref{eq:tight-str-Feller-converge}),
\[
\limsup_{\ell \to \infty} \|w_\ell- w^\circ_\ell - \tilde w\|_\infty \le
\epsilon.
\]
Since this hold for any $\epsilon >0$, $\|w_\ell - w^\circ_\ell - \tilde w\|_\infty \to 0$ as $\ell \to \infty$.

 Assume that  $\cR_0 > 1$ and all
$w^\circ_\ell$ are topologically positive.
By Theorem \ref{re:tight-irred-pandemic}
\begin{equation}
\label{eq:tight-Feller-converge1}
\|\tilde w\|_\infty = \lim_{\ell \to \infty} \|w_\ell - w^\circ_\ell\|_\infty \ge \ln \cR_0, \qquad \ell
\in \N.
\end{equation}
Since $\tilde w$ is continuous and $\kappa$
is topologically irreducible, by Theorem \ref{re:irred-pos2}, $\tilde w$ is strictly positive on $\Omega$.
\end{proof}

\begin{corollary}
\label{re:metric-F-fixedpoint}
Let
$\kappa$ be a strong  Feller kernel that is tight and
topologically irreducible
   and let $\cR_0>1$.
Then there exists some strictly positive
$\tilde w \in C^b_+(\Omega)$
such that $\tilde w = F(\tilde w)$ and $\|\tilde w\|_\infty \ge \ln \cR_0$.
\end{corollary}

\begin{proof}
Apply Theorem \ref{re:tight-Feller-converge} with
$w^\circ_\ell = 1/\ell$.
\end{proof}

\begin{corollary}
\label{re:metric-sandwich-single}
Let
$\kappa$ be a strong  Feller kernel that is tight
and
topologically irreducible
   and let $\cR_0>1$.

Let $w^\circ$ be a topologically positive function
in $ M^b_+(\Omega)$  and $w$ be the minimal
solution of $w =F(w) + w^\circ$.

Then there exist strictly positive $\tilde w \in C^b_+(\Omega)$ with $\|\tilde w\|_\infty \ge \ln \cR_0$
and
\[
\tilde w \le w  - w^\circ.
\]
\end{corollary}

\begin{proof}
Apply Theorem \ref{re:tight-Feller-converge}
with $w^\circ_\ell =(1/\ell) \ w^\circ$.
\end{proof}

%%%%%%%%%%%%%%%%%%%%%%%%%%%%%%%%%%%%%%%%%%%%%%%%%%%%%

\subsection{Comparability kernels}
\label{subsec:compar-ker}
%%%%%%%%%%%%%%%%%%%%%%%%%%%%%%%%%%%%%%%%%%%%%%%%%%

In Corollary \ref{re:metric-sandwich-single},
the fixed point $\tilde w = F(\tilde w)$
may depend on $w^\circ$. The Corollary
would send a much stronger message if this
were not the case.
Further, we do not yet know what happens if $\cR_0=1$.
This leads to the
concept of a comparability kernel (Definition
\ref{def:compar-ker}). Cf. \cite[(K-2)]{Thi79}.

\begin{proposition}
\label{re:kernel-compar-semicon}
Let $\kappa$ be a comparability kernel
and $g \in M^b_+(\Omega)$ be strictly
positive and lower semicontinuous.
Then there exists some $\delta >0$ such
that
\[
\int_\Omega g(\xi) \kappa(d \xi,x) \ge \delta \, \kappa(\Omega,x), \qquad x \in \Omega.
\]
\end{proposition}

\begin{proof}
Let $g \in M^b_+(\Omega)$ be strictly positive
and lower semicontinuous on $\Omega$.
By Lemma \ref{re:top-pos} (b), there exists
a strictly positive Lipschitz continuous $\tilde g \in M^b_+(\Omega)$ such that $0 \le \tilde g \le g$.

Since $\kappa$ is a comparability kernel, there exists some $\delta > 0$ such that
\[
\delta \, \kappa(\Omega,x) \le \int_\Omega \tilde g(\xi) \kappa(d \xi,x) \le
\int_\Omega  g(\xi) \kappa(d \xi,x),
\qquad x \in \Omega. \qedhere
\]
\end{proof}

\begin{proposition}
\label{re:unique-prep}
Assume that the Feller kernel $\kappa$ is
topologically irreducible and a comparability
kernel and $w^\circ\in M_+^b(\Omega)$.
Then, there exists at most one non-zero
 continuous solution $\tilde w \in M^b_+(\Omega)$ to
\[
\tilde w(x) = \int_\Omega f \big (\tilde w(\xi) +w^\circ(\xi)\big) \kappa(d \xi,x), \qquad x \in \Omega.
\]
\end{proposition}

\begin{proof}
Assume that there are two, $\tilde w_1$ and $\tilde w_2$. Since $\kappa$ is topologically irreducible,
by Theorem \ref{re:irred-pos2},
both are strictly positive  and so are $f \circ \tilde w_i$, $i=1,2$. By Proposition \ref{re:kernel-compar-semicon},
there are $\delta_i > 0$ such that
\[
\tilde w_i(x) \ge \delta_i \kappa(\Omega,x), \qquad x \in \Omega.
\]
Since $\tilde w_i \le \kappa(\Omega,\cdot  )$,
$t = \inf_\Omega {\tilde w_1/\tilde w_2} > 0$ and
$\tilde w_1(x) \ge t \tilde w_2(x)$ for all $x \in \Omega$.
Suppose that $t \in (0,1)$.
Since $f$ is increasing,
\[
0\ge \tilde w_1(x) - t \tilde w_2(x)
\ge
\int_\Omega \Big [f\big (t \tilde w_2(\xi) +w^\circ(\xi) \big ) - t f(\tilde w_2(\xi)+w^\circ(\xi)\big)\Big ] \kappa(d \xi,x ).
\]
Since $\tilde w_2 $ is strictly positive and $t \in (0,1)$ and $f$ is increasing and strictly sublinear, for all $\xi \in \Omega$,
\[
\begin{split}
& f\big(t \tilde w_2(\xi) +w^\circ(\xi)\big) - t f\big (\tilde w_2(\xi) +w^\circ(\xi) \big)
\\
\ge &
f\big (t (\tilde w_2(\xi) +w^\circ(\xi))\big ) - t f\big (\tilde w_2(\xi) +w^\circ(\xi)\big)
>0.
\end{split}
\]
Since the left hand side of this inequality
is a continuous function of $\xi$, by
Proposition \ref{re:unique-prep}, there is some
$\delta > 0$ such that
\[
\tilde w_1(x) - t \tilde w_2(x) \ge \delta \kappa(\Omega, x)
\ge \delta \tilde w_2(x),
\qquad
x \in \Omega.
\]
So, $\tilde w_1(x) \ge (t+\delta)\tilde  w_2(x)$ for
all $x \in \Omega$, contradicting the definition
of $t$.

This proves that $\tilde w_1 \ge \tilde w_2$.
By symmetry, equality holds.
\end{proof}

\begin{theorem}
\label{re:unique}
Assume that the Feller kernel $\kappa$ is
topologically irreducible and a comparability
kernel and $w^\circ \in C_+^b(\Omega)$.
Then there exist at most one non-zero
solution $ w \in C^b_+(\Omega)$ to $w = F(w)
+w^\circ$.
\end{theorem}

\begin{proof}
If $w^\circ$ is the zero function, the statement
directly follows from Proposition \ref{re:unique-prep}. If $w^\circ$ is not the
zero function, apply  Proposition \ref{re:unique-prep} to $\tilde w = w - w^\circ \ge 0$.
\end{proof}

\begin{theorem}
\label{re:epino}
Let
$\kappa$ be a  Feller kernel that is
topologically irreducible.
Assume that $\kappa$ is a comparability kernel and $\cR_0 \le 1$.

\begin{itemize}
\item[(a)] Then there exists no nonzero
  solution $\tilde w \in C^b_+(\Omega)$ to
the equation $\tilde w = F(\tilde w)$.

\item[(b)] If, in addition, $\kappa$ is
a strong Feller kernel,
there exists no nonzero
  solution $\tilde w \in M^b_+(\Omega)$ to
the equation $\tilde w = F(\tilde w)$.

\item[(c)] If $\kappa$ is a strong Feller
kernel and $(w^\circ_{\ell})_{\ell \in \N}$ is a  sequence
in $ M^b_+(\Omega)$
and $w^\circ_{\ell} \to 0$ as $\ell \to \infty$
pointwise
on $\Omega$, then $w_\ell \to 0$ as $\ell \to \infty$
pointwise on $\Omega$
for the minimal solutions $w_\ell$ of $w_\ell = F(w_\ell) + w^\circ_{\ell}$.

\end{itemize}
\end{theorem}

\begin{proof}
(a) Assume that such a solution $\tilde w $ exists.
Since $\kappa$ is topologically irreducible,
$w$ is strictly positive on $\Omega$ by Theorem
\ref{re:irred-pos2}.
Let $t =1/2$. Then $f(t\tilde w(x) ) - t f(\tilde w(x)) > 0$ for all $x \in \Omega$ and this difference is a continuous function of $x \in \Omega$. Since
$\kappa$ is a comparability kernel, there is some
$\delta >0$ such that
\[
\begin{split}
\delta \kappa(\Omega, x)
\le &
\int_\Omega [f(t\tilde w(\xi) ) - t f(\tilde w(\xi))] \kappa(d \xi, x )
\\
= &
F (t \tilde w)(x) - t F(\tilde w)(x), \qquad
x \in \Omega.
\end{split}
\]
So there exists some $\epsilon > 0 $ such that
\[
F(t \tilde w) \ge (1+\epsilon) t F(\tilde w )
= (1+\epsilon) t \tilde w.
\]
By (\ref{eq:order-der}),
$
K \tilde w \ge (1+\epsilon) \tilde w.
$
This implies $\cR_0 = \br(K) \ge (1+\epsilon)$
\cite[Thm.3.1]{Thi17}\cite[Thm5.31]{Thi24},
a contradiction.

(b) If $F(\tilde w) =w \in M^b_+(\Omega)$,
and $\kappa$ is a strong Feller kernel,
then $\tilde w$ is continuous.

(c) This follows from part (b)
and Corollary \ref{re:final-size-converge-zero}.
\end{proof}

\begin{theorem}
\label{re:ultimate}
Let
$\kappa$ be a strong  Feller kernel that is tight and
topologically irreducible.
Assume that $\kappa$ is a comparability kernel and $\cR_0>1$.

\begin{itemize}

\item[(a)] Then there exists a unique
nonzero solution $\tilde w \in  M^b_+(\Omega)$ to
the equation $\tilde w = F(\tilde w)$; $\tilde w$
is continuous, strictly positive and $ \|\tilde w \| \ge \ln \cR_0$.

\item[(b)] For any topologically positive $w^\circ$ in $ M^b(\Omega)$ and the minimal solution $w \in M^b_+(\Omega)$
of $w = F(w) + w^\circ$, we have $w -w^\circ\ge \tilde w$
with the unique $\tilde w$ from (b).

\item[(c)] Finally, if $(w^\circ_{\ell})_{\ell \in \N}$ is a  decreasing sequence
of topologically positive functions in $ M^b_+(\Omega)$ with  $w^\circ_{\ell} \to 0$ as $\ell \to \infty$
uniformly on all compact subsets
of $\Omega$, then $w_\ell -w^\circ_\ell \to \tilde w$ uniformly on $\Omega$
for the solutions $w_\ell$ of $w_\ell = F(w_\ell) + w^\circ_{\ell}$.
\end{itemize}
\end{theorem}

\begin{proof}
 (a) Since $\kappa$ is a strong
 Feller kernel, any solution $\tilde w =F(\tilde w)$ in $M^b_+(\Omega)$ is continuous.
 Uniqueness now follows from Theorem \ref{re:unique},
existence from
Corollary \ref{re:metric-F-fixedpoint}.

(b) This follows from part (a) and
Corollary \ref{re:metric-sandwich-single}.

(c) Combine Theorem \ref{re:tight-Feller-converge}
and parts (a) and (b).
\end{proof}

%%%%%%%%%%%%%%%%%%%%%%%%%%%%%%%%%%%%%%%%%%%%%%%%%%%%

\subsubsection{Examples of comparability kernels}
\label{subsubsec:compar-ker-exp}

\begin{proposition}
\label{re:comparab1}
A Feller  $\kappa$ is a comparability  kernel if
there exist some compact subset $W$ of $\Omega$
and some $\delta > 0$ such that
\[
\kappa(W, x) \ge \delta \,\kappa(\Omega, x),
\qquad
x \in \Omega.
\]
\end{proposition}

\begin{proof}
Let $g : \Omega \to (0,\infty)$ be continuous.
Let $W$ and $\delta >0$ as in the statement of the proposition. Since $W$ is compact, $\inf_W g> 0$.
Further, for all $x \in \Omega$,
\[
Kg (x) \ge \int_W g(\xi) \kappa(d \xi,x)
\ge
\inf_w g \;\kappa(W,x)
\ge
\inf_W g \;\delta \, \kappa(\Omega,x). \qedhere
\]
\end{proof}

\begin{proposition}
\label{re:comparab2}
Let $\kappa$ be a tight Feller kernel
and
\begin{equation}
\inf_{x \in \Omega} \kappa(\Omega, x) > 0.
\end{equation}
Then $\kappa$ is a comparability kernel.
\end{proposition}

\begin{proof}
By assumption,
$\delta = \inf_{x \in \Omega} \kappa(\Omega, x) > 0$.
Since $\kappa$ is tight, there exists a
compact subset $W$ of $\Omega$ such that
$
\kappa(\Omega \setminus W,x ) \le \delta/2$
for $ x \in \Omega$.
Since $\kappa(\Omega, \cdot)$ is bounded on
$\Omega$, the assertion follows from
Proposition \ref{re:comparab1}.
\end{proof}

\begin{proposition}
\label{re:separ-envelop}
Let $n \in \N$ and $\nu_j, \mu_j \in \cM_+(\Omega)$ be nonzero measures and $k_j \in M^b_+(\Omega)$ be nonzero functions,
$j=1, \ldots, n$, such that
\begin{equation}
\label{eq:separ-envelop}
\sum_{i=1}^n \mu_i(\omega) k_i(x)\le \kappa(\omega, x ) \le \sum_{i=1}^n \nu_i(\omega) k_i(x),
\qquad
x \in \Omega, \quad \omega \in \cB.
\end{equation}
Then $\kappa$ is a  comparability kernel.
Further, $\kappa$ is tight if
all $\nu_i$ are tight measures, e.g, if $\Omega$
is a Polish space (complete and separable) .
\end{proposition}

Compare \cite[Exp.1.3b]{Thi79}. It follows that
every semi-separable Feller kernel, (\ref{eq:separ-intro}), is a
comparability kernel.

\begin{proof}
By (\ref{eq:separ-envelop}),
\[
\kappa(\Omega, x)
\le
\sum_{i=1}^n \nu_i(\Omega) k_i(x)
\le \sup_{i=1}^n \nu_i(\Omega) \sum_{i=1}^n k_i(x),
\qquad
x \in \Omega.
\]
Let $g \in M^b_+(\Omega)$ be strictly positive.
By (\ref{eq:separ-envelop}),
\[
\int_\Omega g(\xi) \kappa(d\xi,x)
\ge
\sum_{i=1}^n \Big (\int_\Omega g \, d \mu_i \Big )
k_i(x).
\]
Since $g$ is strictly positive on $\Omega$,
$
\Omega = \bigcup_{\ell \in \N} \{g \ge 1/\ell\}.
$
Since $\mu_i \in  \cM_+(\Omega)$,
\[
0 < \mu_i(\Omega) = \lim_{\ell \to \infty}
\mu_i \big (\{g \ge 1/\ell\}\big ).
\]
For sufficiently large $\ell \in \N$,
\[
\int_\Omega g d \mu_i \ge \frac{1}{\ell}
\mu_i \big (\{g \ge 1/\ell\}\big ) > 0.
\]
We combine these inequalities: For
all $x\in \Omega$,
\[
\int_\Omega g(\xi) \kappa(d \xi,x)
\ge
\inf_{i=1}^n \int_\Omega g d \mu_i
\sum_{i=1}^n k_i(x)
\ge
\frac{\displaystyle\inf_{i=1}^n \int_\Omega g d \mu_i}
{\displaystyle \sup_{i=1}^n \nu_i(\Omega)}
\kappa ( \Omega,x). \qedhere
\]
\end{proof}

%%%%%%%%%%%%%%%%%%%%%%%%%%%%%%%%%%%%%%%%%%%%%%%%%%%%
\bigskip

\noindent
{\bf Acknowledgement.}
I thank Odo Diekmann, Hisashi Inaba and Jordi Ripoli for
valuable comments. In particular,  Odo Diekmann pointed out a missing assumption
in Remark \ref{re:kernel-final-susc-intro}.
See an earlier version of this paper that was
uploaded on arXiv.

%%%%%%%%%%%%%%%%%%%%%%%%%%%%%%%%%%%%%%%%%%%%%%%%%%%%%

\bibliography{}

%%%%%%%%%%%%%%%%%%%%%%%%%%%%%%%%%%%%%%%%%%%%%%%%%%%%%%%%%%%%%
%\bibliographystyle

\end{document}